\documentclass[a4paper]{siamart220329}

\usepackage{amsmath, latexsym, amsfonts, amssymb, amscd,bbm}
\usepackage[english]{babel}
\usepackage{mathrsfs,stmaryrd}

\usepackage{color}

\newsiamthm{hypothesis}{Hypothesis}
\newsiamthm{claim}{Claim}
\newsiamremark{remark}{Remark}
\crefname{hypothesis}{Hypothesis}{Hypotheses}

\newcommand{\norm}[1]{\left \| #1 \right \|}

 \newcommand{\lsup}[1]{\underset{#1\to\infty}{\overline{\lim}}}

\title{The Hydrodynamic Limit of Neural Networks with Balanced Excitation and Inhibition}
\author{  James MacLaurin, Pedro Vilanova }

\begin{document}
\maketitle

\begin{abstract}
The theory of `Balanced Neural Networks' is a very popular explanation for the high degree of variability and stochasticity in the brain's activity. We determine equations for the hydrodynamic limit of a balanced all-to-all network of $2n$ neurons for asymptotically large $n$. The neurons are divided into two classes (excitatory and inhibitory). Each excitatory neuron excites every other neuron, and each inhibitory neuron inhibits all of the other neurons. The model is of a stochastic hybrid nature, such that the synaptic response of each neuron is governed by an ordinary differential equation. The effect of neuron $j$ on neuron $k$ is dictated by a spiking Poisson Process, with intensity given by a sigmoidal function of the synaptic potentiation of neuron $j$. The interactions are scaled by $O(n^{-1/2})$, which is much stronger than the $O(n^{-1})$ scaling of classical interacting particle systems. We demonstrate that, under suitable conditions, the system does not blow up as $n\to \infty$ because the network activity is balanced between excitatory and inhibitory inputs. The limiting population dynamics is proved to be Gaussian: with the mean determined by the balanced between excitation and inhibition, and the variance determined by the Central Limit Theorem for inhomogeneous Poisson Processes. The limiting equations can thus be expressed as autonomous Ordinary Differential Equations for the means and variances.
\end{abstract}

\section{Introduction}

In theoretical neuroscience, it is widely conjectured that neurons are typically dynamically balanced, with a high number of excitatory and inhibitory inputs \cite{Shadlen1994,Tsodyks1995,Amit1997,Vreeswijk1996,VanVreeswijk1998,Shadlen1998,Engelken2022}. It is thought that the dynamic balance could explain the high degree of stochasticity and variability in cortical discharge, which is indicated by the fact that the coefficient of variation in cortical spike trains is typically $O(1)$ \cite{Softky1993}. Roughly speaking, the theory is that the mean excitation and inhibition approximately `cancel', and what is left are the stochastic fluctuations about the mean \cite{Shadlen1994,Amit1997}. This balanced paradigm has proved extremely popular and has been explored in numerous directions. Some applications include: explaining oscillations and rhythms in brain activity \cite{Brunel2000}, UP / DOWN transitions \cite{Tartaglia2017}, working memory models \cite{BrunelWang2001Working,Lim2014,Rubin2017}, pattern formation \cite{Rosenbaum2014} and spatially-distributed neural activity \cite{Litwin-Kumar2012,Rosenbaum2017}.

We consider the following simple `balanced stochastic network' model, motivated by the model of Van Vreeswijk and Sompolinsky \cite{Vreeswijk1996,VanVreeswijk1998}. The model consists of two types of neurons: $n \gg 1$ excitatory neurons, and $n \gg 1$ inhibitory neurons. Each neuron is connected to every other neuron. Excitatory neurons excite all of the other neurons and inhibitory neurons inhibit all of the other neurons. Although Sompolinsky and Van Vreeswijk considered sparse networks, we limit ourselves to studying all-to-all networks. The model is of a hybrid nature (also termed a Piecewise Deterministic Markov Process \cite{Bressloff2018c,Bressloff2021a}): the spiking of neurons is Poissonian and stochastic (representing the popular view that the dominant source of stochasticity in neurons is spike-transmission failure \cite{Manwani1999}), and the synaptic response is given by a linear ordinary differential equation \cite{Mcdonnell2017}. Note that in most high-dimensional neuronal network models, there are $O(n)$ sources of noise (whether Poissonian or Brownian), i.e. the noise is directly afferent on the neuron \cite{Baladron2012,Chevallier2017,Locherbach2018,Duval2022}. However in our model, there are $O(n^2)$ sources of noise, i.e. the noise is afferent on the synaptic connections, rather than the neurons themselves.

The scaling of the interaction is different from the standard $n^{-1}$ scaling for particle systems with weak interactions (i.e. for Mckean-Vlasov systems \cite{Sznitman1989,Baladron2012,DeMasi2014,Jabin2017,Lucon2020a,Chaintron2021,Bramburger2023,Avitabile2024} or high-dimensional Poissonian chemical reaction networks \cite{Anderson2015,Chevallier2017,Agazzi2018,AgatheNerine2022,Avitabile2024_2}). The scaling factor is $n^{-1/2}$ (so the effect of one neuron on another is relatively stronger than in the Mckean-Vlasov case). The reason that the dynamics does not blowup is that each neuron receives a balance of excitatory and inhibitory inputs which approximately cancel out. This strong scaling requires the development of new techniques to make sense of the high-dimensional limit.

To obtain a hydrodynamic limit, we make use of the fact that the `balanced state' is strongly attracting. This effectively damps down the $O(n^{-1/2})$ terms that could cause blowup, and we are left with the fluctuations about the mean \cite{Katzenberger1991}. See for instance the textbook by Berglund and Gentz for an overview of methods for studying stochastic systems near strongly-attracted manifolds \cite{Berglund2006}, or the recent papers \cite{Parsons2017a,Adams2025} that study the quasi-steady distribution for stochastic systems near attracting manifolds over long timescales. We find that the empirical measure of the system concentrates at a unique Gaussian limit as $n\to\infty$: the variance of the system evolves analogously to the Ornstein-Uhlenbeck stochastic process, and the mean activity is determined by the strong interactions pulling the system towards a balance of excitation and inhibition.

There are several other recent works that have explored the effect of inhibition on mean-field interacting particle systems. Erny, Locherbach and Loukianova consider interacting Hawkes Processes in the diffusive regime (so that interactions are scaled by $(n^{-/2})$). In their work, the noise is directly afferent on the neurons, and they find that in the large $n$ limit the average neural activity does not concentrate but has Gaussian fluctuations \cite{Erny2021}. Pfaffelhuber, Rotter and Stiefel \cite{Pfaffelhuber2022} consider a system of Hawkes Process in the balanced regime; the chief difference from our paper is that in their paper the noise is afferent on the neurons, rather than the synaptic connections. Duval, Lucon and Pouzat \cite{Duval2022} consider a network of Hawkes Processes with multiplicative inhibition. Their work differs from this one insofar as the noise is directly afferent on the neurons, and the effect of one neuron on another scales as $O(n^{-1})$.

It is also worth comparing these equations to the `spin-glass' dynamical models \cite{Sompolinsky1981a,Crisanti1993,BenArous1995,BenArous2006,Faugeras2015,Crisanti2018,Helias2019,MacLaurin2024}. These models also have $n^{-1/2}$ scaling of the interactions. However, the interactions themselves are mediated by static Gaussian random variables, of zero mean and unit variance. These models have also been heavily applied to neuroscience \cite{Moynot2002,Faugeras2015,Parisi2023}. One of the most important differences is that in the spin glass model, an individual neuron has both excitatory and inhibitory effects on other neurons. However, in the balanced model of this paper, individual neurons are either purely excitatory, or purely inhibitory. Indeed, the hydrodynamic limiting equations are different (one can compare the limiting equations of this paper to for instance the equations in \cite{Crisanti2018,MacLaurin2024}).

\section{Model Outline}

For $\beta \in \lbrace e,i \rbrace$ (i.e. `excitatory' or `inhibitory') and $j\in I_n := \lbrace 1,2,\ldots, n \rbrace$, the synaptic dynamics of the $(\beta,j)$ neuron is indicated by the state variable $u^j_{\beta,t} \in \mathbb{R}$. We let $Z^{jk}_{\alpha\beta}(t) \in \mathbb{N}$ count the number of spikes from neuron $ (\beta,k) \to (\alpha,j) $ up to time $t$. The synaptic dynamics is taken to be linear, assuming the form
\begin{align}
u^{j}_{\alpha,t} &= u^j_{\alpha,0} + \int_0^t -  \tau_{\alpha}^{-1} u^{j}_{\alpha,s} ds + C_{\alpha e} n^{-1/2} \sum_{k\in I_n}Z^{jk}_{\alpha e}(t) - C_{\alpha i} n^{-1/2} \sum_{k\in I_n}Z^{jk}_{  \alpha i}(t) 
\end{align}
where $ \lbrace C_{ \alpha\beta} \rbrace_{\alpha,\beta \in \lbrace e,i\rbrace }$ are non-negative constants indicating the relative effects of excitation and inhibition. 

The system is Poissonian \cite{Bremaud2020}. It is thus assumed that there are propensity functions $f_{\alpha\beta}: \mathbb{R} \to \mathbb{R}_{\geq 0}$ such that the probability that there is a spike on the $(\alpha,k) \mapsto (\beta,j)$ synapse (for $\alpha,\beta  \in \lbrace e,i \rbrace$) over the time interval $[t, t+h]$ ($h \ll 1$) is
\begin{align}
h f_{\beta\alpha}( u^{k}_{\alpha,t}) + O(h^2).
\end{align}
The solution may be represented using the following time-rescaled formulation \cite{Anderson2015}. Let \newline $\lbrace Y^{jk}_{\alpha\beta}(t) \rbrace_{j,k\leq n \fatsemi \alpha,\beta\in \lbrace e,i\rbrace}$ be independent Poisson counting processes, and define
\begin{align}
Z^{jk}_{\alpha \beta}(t) = Y^{jk}_{\alpha\beta}\bigg( \int_0^t f_{\alpha\beta}(u^{k}_{\beta,s}) ds \bigg) 
\end{align}
We underscore the fact that, in this model, the stochasticity is afferent on the synaptic connections, rather than the neurons themselves. This corresponds to the view that the dominant source of stochasticity in high-dimensional neural networks is synaptic transmission failure.

\subsection{Notation}

The topology on the space of measures is always that of weak convergence \cite{Billingsley1999}. In other words, a sequence $\mu_n$ converges to $\mu$ if and only if for all bounded continuous functions $f$,
\[
\mathbb{E}^{\mu_n}[f] \to \mathbb{E}^{\mu}[f].
\]
We endow $\mathcal{C}([0,t],\mathbb{R})$ with the topology generated by the supremum norm,
\begin{align}
\norm{x}_t = \sup_{s\leq t}\big| x_s \big|.
\end{align}
For two probability measures $\mu,\nu \in \mathcal{P}\big( \mathcal{C}([0,t],\mathbb{R}) \big)$, the Wasserstein Distance is written as
\begin{align}
d(\mu , \nu ) = \inf_{\xi}\bigg\lbrace \mathbb{E}^{\xi}\big[ \norm{x-y}_t \big] \bigg\rbrace
\end{align}
where the infimum is over all probability measures $\xi \in \mathcal{P}\big( \mathcal{C}([0,t],\mathbb{R})^2 \big)$, such that the law of $x$ is $\mu$, and the law of $y$ is $\nu$.
\section{Assumptions}

We start by making some basic assumptions on the regularity of the propensity functions.

\begin{hypothesis}
It is assumed that for all $\alpha,\beta \in \lbrace e,i\rbrace$, $z \mapsto f_{\alpha\beta}(z)$ are globally Lipschitz functions. It is assumed that there is a constant $C_f $ such that for all $x \in \mathbb{R}$ and all $\alpha,\beta \in \lbrace e,i\rbrace$,
\begin{align}
0 < f_{\alpha\beta}(x) \leq C_f.
\end{align}
\end{hypothesis}
Typically the intensity functions are sigmoidal (smooth, bounded, and monotonically increasing), but we do not necessarily need to make this assumption \cite{Ermentrout2010}.

Thus far, we have not yet made sufficient assumptions that ensure that the system does not blow up in finite time as $n\to\infty$. One option would be to make various assumptions on the propensity functions that guarantee that for all time, the synaptic activity is with high probability below some threshold independent of $n$ (in fact we do this in Section ). However, for the moment we avoid making stringent assumptions on the propensity functions, and instead we seek to determine which parts of the phase space are stable (in the limit as $n\to\infty$). We term this part of the phase space the `balanced manifold'. Our main theorem concerns the prediction of the dynamics as long as it stays near the balanced manifold (predicting the dynamics away from the balanced manifold is more difficult and will be pursued in a subsequent paper).

Intuitively, the balanced manifold is such that (i) the net excitatory inputs throughout the system are almost identical to the net inhibitory inputs throughout the network, and (ii) the manifold is locally stable (as long as the fluctuations are Gaussian). To this end, let us first define the auxiliary functions
\begin{align}
F_e, F_i  :& \mathbb{R}^2 \times \mathbb{R}^+ \times \mathbb{R}^+ \mapsto \mathbb{R} \\
F_e (v_{e} , v_i,  K_e ,  K_i ) =&  \int_{\mathbb{R}}\big( C_{ee} \rho(K_e,x)  f_{ee} ( v_e + x ) - C_{ei} \rho(K_i,x)   f_{ei}( v_i + x)  \big) dx  \\
F_i (v_{e} , v_i,  K_e ,  K_i ) =&  \int_{\mathbb{R}}\big( C_{ie} \rho(K_e,x)  f_{ie} ( v_e + x ) - C_{ii} \rho(K_i,x)   f_{ii}( v_i + x)  \big) dx   \text{ where } \\
\rho(K,x) =& \big( 2\pi K \big)^{-1/2} \exp\bigg( -x^2 / (2K) \bigg), \label{eq: rho K x definition}
\end{align}
and notice that $\rho(K,x)$ is the density of a centered Gaussian distribution with variance $K$. We next define the Jacobian to be such that
\begin{align} \label{eq: J v definition}
 J_v(v_e , v_i ,  K_e, K_i) = \left( \begin{array}{c c}
 \partial_{v_e} F_e(v_e , v_i ,  K_e, K_i) &  \partial_{v_i} F_e(v_e , v_i ,  K_e, K_i) \\
  \partial_{v_e} F_i(v_e , v_i ,  K_e, K_i) &  \partial_{v_i} F_i(v_e , v_i ,  K_e, K_i) 
 \end{array}\right)
\end{align}
Define the balanced manifold to be
\begin{align}
\mathcal{U} &\subseteq \mathbb{R} \times \mathbb{R} \times  \mathbb{R}^+  \times \mathbb{R}^+  \\
\mathcal{U} &= \bigg\lbrace (v_e, v_i ,  K_e , K_i)  \; : F_e (v_e, v_i ,  K_e , K_i)  =0 \text{ and } F_i(v_e, v_i ,  K_e , K_i) =0 \text{ and }\nonumber \\ & \text{ The real parts of the eigenvalues of }J_v(v_e, v_i, K_e , K_i)  \text{ are strictly negative.}    \bigg\rbrace .
\end{align}
Our next assumption is that the initial conditions concentrate on the balanced manifold as $n\to\infty$. In more detail,
\begin{hypothesis} \label{Hypothesis Initial Conditions on the Balanced Manifold}
Write the empirical measure of initial conditions
\begin{align}
\hat{\mu}^n_e = n^{-1}\sum_{j\in  I_n} \delta_{u^j_{e,0}} \in \mathcal{P}(\mathbb{R})\; \; \text{ and } \; \; \hat{\mu}^n_i = n^{-1}\sum_{j\in  I_n} \delta_{u^j_{i,0}} \in \mathcal{P}(\mathbb{R}).
\end{align}
It is assumed that the initial conditions are such that there exist Gaussian measures $\mu_{e,ini}, \mu_{i,ini} \in \mathcal{P}\big( \mathbb{R} \big)$ such that
\begin{align}
\lim_{n\to\infty} \hat{\mu}^n_e &= \mu_{e,ini} \\
\lim_{n\to\infty} \hat{\mu}^n_i &= \mu_{i,ini}  .
\end{align}
Write $m_e , m_i \in \mathbb{R}$ to be the means, and $V_e, V_i \in \mathbb{R}^+$ to be the respective covariances. It is assumed that the initial conditions enforce balanced excitation and inhibition, i.e. $(m_e, m_i, V_e, V_i) \in \mathcal{U}$.
\end{hypothesis}

\section{Results}

We define the (empirical) average levels of excitation and inhibition to be
\begin{align}
\hat{v}^n_e(t) =& n^{-1}\sum_{j\in I_n} u_e(t) \\
\hat{v}^n_i(t) =& n^{-1}\sum_{j\in I_n} u_i(t) .
\end{align}
We similarly define the empirical covariances to be
\begin{align}
\hat{K}^n_e(t) =& n^{-1}\sum_{j\in I_n} \big(u^j_e(t) - \hat{v}^n_e(t) \big)^2 \\
\hat{K}^n_i(t) =& n^{-1}\sum_{j\in I_n} \big(u^j_i(t) - \hat{v}^n_i(t) \big)^2 .
\end{align}
We are going to define $\eta \in (0,\infty]$ to be the time that (in the hydrodynamic limit) the system first leaves the balanced manifold. Our main result (in the following theorem) is that for all times less than or equal to $\eta$, the hydrodynamic limit exists and can be precisely specified.
\begin{theorem} \label{Big Theorem}
With unit probability, for any $T < \eta$, 
\begin{align}
\lim_{n\to\infty} \sup_{t\leq T}\bigg\lbrace   \big| \hat{K}^n_e(t) - K_e(t) \big| + \big| \hat{K}^n_i(t) - K_i(t) \big| + \big| \bar{v}_e(t) - \hat{v}^n_e(t) \big| + \big| \bar{v}_i(t) - \hat{v}^n_i(t) \big|   \bigg\rbrace  = 0.
\end{align}
where $(\bar{v}_e,\bar{v}_i,K_e,K_i)$ and $\eta$ are defined (below) in Lemma \ref{Hypothesis Manifold Evolution}. 
\end{theorem}

In the following lemma, we outline the equations that must be satisfied by the mean and variance in the large size limit.
\begin{lemma} \label{Hypothesis Manifold Evolution}
Let $\eta \geq 0$ be the largest number (possibly infinite) such that (i) there exists a unique $(\bar{v}_e , \bar{v}_i,  K_e , K_i ) \in  \mathcal{C}^1\big( [0,\eta], \mathbb{R}^2 \times \mathbb{R}_{\geq 0}^2\big)$ such that for all $t < \eta$
\begin{align}
\frac{dK_e}{dt} =& -2 K_e / \tau_e +  \Sigma_e(\bar{v}_e(t), \bar{v}_i(t) ,K_e(t) , K_i(t) ) \label{eq: invariant covariance 1} \\
\frac{dK_i}{dt} =& - 2 K_i /  \tau_i  +  \Sigma_i(\bar{v}_e(t), \bar{v}_i(t),K_e(t) , K_i(t) ) \label{eq: invariant covariance 2} \\
 K_e (0) =& k_e  \\
  K_i (0) =& k_i  \\
F_e\big(\bar{v}_e(t), \bar{v}_i(t), K_e(t) , K_i(t) \big) =& 0 \label{eq: Fe zero equation} \\
F_i\big(\bar{v}_e(t), \bar{v}_i(t), K_e(t) , K_i(t) \big) =& 0  \label{eq: Fi zero equation}. \\
\bar{v}_e(0) =& m_e \\
\bar{v}_i(0) =& m_i.
\end{align}
where
\begin{align}
\Sigma_{e}(\bar{v}_e,\bar{v}_i, K_e, K_i) &=   \int_{\mathbb{R}}\big(  \rho(K_{e},x) C^2_{ee} f_{ee} ( \bar{v}_e + x) + \rho(K_{i},x)  C^2_{ei} f_{ei}(\bar{v}_i + x)  \big) dx \\
\Sigma_{i}(\bar{v}_e, \bar{v}_i ,K_e,K_i) &=   \int_{\mathbb{R}}\big(  \rho(K_{e},x) C^2_{ie} f_{ie} ( \bar{v}_e + x) + \rho(K_{i},x)  C^2_{ii} f_{ii}( \bar{v}_i + x)  \big) dx \text{ and }
\end{align}
$\rho$ is defined in \eqref{eq: rho K x definition} and for all $t < \eta$,
 \begin{align}
  (\bar{v}_e(t), \bar{v}_i(t) ,K_e(t) , K_i(t) ) \in \mathcal{U}.
 \end{align}
\end{lemma}
For $t < \eta$, let $\zeta_t < 0$ be the largest real component of the eigenvalues of $J_v\big(\bar{v}_e(t), \bar{v}_i(t), K_e(t) , K_i(t) \big)$ (defined in \eqref{eq: J v definition}).

\section{Proof Overview}

We fix some $T < \eta$ for the rest of this paper. We are going to see that the probability distribution of the empirical measure $\hat{\hat{\mu}}^n_{\alpha} := n^{-1}\sum_{j\in I_n} \delta_{u^j_{\alpha,[0,T]}}$ ( for $\alpha \in \lbrace e,i \rbrace$) becomes Gaussian in the large $n$ limit. To this end, to facilitate the proofs, it is convenient to split the dynamics into (i) a `system-wide mean activity' and (ii) `local fluctuations about the mean'; a decomposition which is straightforward thanks to the fact that the synaptic dynamics is linear. In more detail, we make the decomposition, for $\alpha \in \lbrace e,i \rbrace$, 
\begin{equation}
u^k_{\alpha,t} = v^n_{\alpha}(t) + x^k_{\alpha,t}
\end{equation}
where  $v^n_{\alpha}(t)$ is such that
\begin{align}
v^n_{\alpha}(t) &= v^n_{\alpha}(0) +  \int_0^t \bigg( -  \tau_{\alpha}^{-1} v^n_{\alpha}(s)  + C_{\alpha e} n^{-1/2} \sum_{k\in I_n}f_{\alpha e}\big( u^k_{e,s} \big) - C_{\alpha i} n^{-1/2} \sum_{k\in I_n}f_{  \alpha i}\big(u^k_{i,s} \big) \bigg) ds .
\end{align}
and
\begin{align}
x^j_{\alpha,t} &= u^j_{\alpha,0} - v^n_{\alpha}(0) - \tau_{\alpha}^{-1} \int_0^t x^j_{\alpha,s} ds + W^j_{\alpha}(t) \label{eq: x j dynamics} \\
 W^j_{\alpha }(t) &=   n^{-1/2} \sum_{k\in I_n}\bigg( C_{\alpha e } Z^{jk}_{\alpha e}(t) - C_{\alpha i} Z^{jk}_{\alpha i}(t)   -  \int_0^t \big\lbrace  C_{\alpha e}  f_{\alpha e} \big(u^{k}_{e,s}\big) - C_{ \alpha i}  f_{\alpha i} \big(u^{k}_{i,s}\big) \big\rbrace  ds  \bigg).
 \end{align}
Furthermore, thanks to the time-rescaled representation of Poisson Processes \cite{Anderson2015}, we can make the representation
\begin{multline}
W^j_{\alpha}(t) =  C_{\alpha e} n^{-1/2} Y^j_{\alpha e}\bigg( \sum_{k\in I_n} \int_0^t  f_{\alpha e} \big(u^{k}_{e,s}\big) ds \bigg)
- C_{\alpha i} n^{-1/2} Y^j_{\alpha i}\bigg( \sum_{k\in I_n} \int_0^t  f_{\alpha i} \big(u^{k}_{i,s}\big) ds \bigg) \\
 - n^{-1/2}\sum_{k\in I_n}  \int_0^t \big\lbrace  C_{\alpha e}  f_{\alpha e} \big(u^{k}_{e,s}\big) - C_{\alpha i}  f_{\alpha i} \big(u^{k}_{i,s}\big) \big\rbrace  ds ,
\end{multline}
where $\lbrace Y^j_{\alpha\beta}(t) \rbrace_{j\in I_n \fatsemi \alpha,\beta \in \lbrace e,i \rbrace, t\geq 0}$ are independent unit-intensity counting processes. Write
\begin{align}
w^j_{\alpha\beta}(t) = n^{-1/2} Y^j_{\alpha\beta}(nt) - \sqrt{n} t.
\end{align}
Writing
\begin{align}
 Q^n_{\alpha\beta}(t) = n^{-1} \sum_{k\in I_n} f_{\alpha\beta}(u^k_{\beta,t})  \text{ and }
 \Lambda^n_{\alpha\beta}(t) = \int_0^t Q^n_{\alpha\beta}(s) ds,
 \end{align}
 we thus have that
 \begin{equation}
 W^j_{\alpha}(t) =   C_{\alpha e} w^j_{\alpha e}\big(  \Lambda^n_{\alpha e}(t) \big) - C_{\alpha i} w^j_{\alpha i}\big(  \Lambda^n_{\alpha i}(t) \big). 
 \end{equation} 
 For $\alpha \in \lbrace e,i \rbrace$, let $\nu^n_{\alpha} \in \mathcal{P}\big( \mathcal{C}([0,T], \mathbb{R})  \big)$ be the law of a centered Gaussian variable $x_{\alpha} \in \mathcal{C}([0,T], \mathbb{R}) $, defined as follows. Let $\big\lbrace W_{\alpha\beta}(t) \big\rbrace_{\alpha,\beta\in \lbrace e,i \rbrace}$ be independent Brownian Motions \newline(independent of $\lbrace Q^n_{\alpha\beta}(s) \rbrace_{\alpha,\beta \in \lbrace e,i \rbrace, s\leq T}$). Define $\nu^n_{\alpha}$ to be the law of $\big( x_{\alpha}(t) \big)_{t\leq T}$, where
\begin{equation}
x_{\alpha}(t) = x_{\alpha}(0) - \tau_{\alpha}^{-1} \int_0^t x_{\alpha}(s) ds + C_{\alpha e} W_{\alpha e}\bigg( \int_0^t Q^n_{\alpha e}(s) ds \bigg) - C_{\alpha i} W_{\alpha i }\bigg( \int_0^t Q^n_{\alpha i}(s) ds \bigg)
\end{equation}
Standard Gaussian arithmetic implies that, conditionally on $\lbrace Q^n_{\alpha\beta}(s) \rbrace_{s\leq T}$, $\nu^n_{\alpha}$ is a centered Gaussian random variable \cite[Section 5.6]{Karatzas1991}. Let us underscore the fact that $\nu^n_{\alpha}$ is itself a random variable because its covariance function is random.

Write
\begin{align}
L^n_t = \sup_{r,s \leq t} \big| \mathbb{E}^{\nu^n_e}[ x_r x_s ] - \mathbb{E}^{\mu_e}[x_r x_s] \big| + \sup_{r,s \leq t} \big| \mathbb{E}^{\nu^n_i}[ x_r x_s ] - \mathbb{E}^{\mu_i}[x_r x_s] \big|.
\end{align}

\begin{lemma} \label{Lemma bound v n e t v n i t}
There exists a constant $\tilde{C}_T > 0$ such that for all $n\geq 1$,
\begin{align}
\sup_{t\leq T} \big\lbrace | v^n_{e}(t) -  \bar{v}_e(t) |
+ | v^n_{i}(t) -  \bar{v}_{i}(t) | + L^n_t \big\rbrace \leq \tilde{C}_T n^{-1/2} +  \tilde{C}_T d_W\big(\hat{\mu}^n_{e,t} , \nu^n_{e,t} \big) + \tilde{C}_T d_W\big(\hat{\mu}^n_{i,t} , \nu^n_{i,t} \big)
\end{align}
\end{lemma}

Next, we are able to demonstrate that the empirical measures converge to be Gaussian in the large $n$ limit.
\begin{lemma}\label{Lemma Convergence to Gaussian}
For any $\epsilon > 0$, 
\begin{equation}
\lsup{n} n^{-1} \log \mathbb{P}\bigg( \text{For an }\alpha \in \lbrace e,i\rbrace , \; \; d_W\big( \hat{\mu}^n_{\alpha} , \nu_{\alpha} \big) \geq \epsilon  \bigg) < 0
\end{equation}
\end{lemma}
An immediate corollary is that the `average fluctuations' about the mean are indeed zero.
\begin{corollary}\label{Corollary Empirical Mean}
For any $\epsilon > 0$,
\begin{equation}
\lsup{n}n^{-1} \log \mathbb{P}\bigg( n^{-1} \sup_{t\leq T}\sum_{j\in I_n} x^j_{e,t} > \epsilon \text{ or }n^{-1} \sup_{t\leq T}\sum_{j\in I_n} x^j_{i,t} > \epsilon \bigg) < 0.
\end{equation}
\end{corollary}
\begin{proof}
Notice that
\[
n^{-1} \sum_{j\in I_n} x^j_{\alpha,t} = \mathbb{E}^{\hat{\mu}^n_{\alpha}}[ x_t ].
\]
Since $\mathbb{E}^{\mu_{\alpha}}[x_t] = 0$ for all $t\leq T$, the Lemma is a consequence of Lemma \ref{Lemma Convergence to Gaussian}.
\end{proof}

Next, define
\begin{align}
q^n_{\alpha}(t) = v^n_{\alpha}(t) -  \bar{v}_{\alpha}(t) .
\end{align}
Our next result uses the fact that $\mathcal{U}$ is strongly attracting to bound the difference of the system-wide mean from its large $n$ limit. Let $\mu_{e,t} , \mu_{i,t} \in \mathcal{P}(\mathbb{R})$ each be centered Gaussian measures, with (respective) variances $K_e(t)$ and $K_i(t)$ (the limiting covariances specified in Lemma \ref{Hypothesis Manifold Evolution}). 
\begin{lemma}\label{Lemma bound v norm difference}
There exist constants $c_T , \epsilon_T  > 0$  such that as long as 
\begin{equation}
\sup_{t\leq T} \big\lbrace  d_W\big(\hat{\mu}^n_{e,t} , \mu_{e,t} \big) + d_W\big(\hat{\mu}^n_{i,t} , \mu_{i,t} \big) \big\rbrace \leq \epsilon_T
\end{equation}
then for all $t\leq T$,
\begin{align}
\sup_{s\leq t} \big\lbrace q^n_e(s)^2 + q^n_i(s)^2 \big\rbrace^{1/2} \leq  c_T n^{-1/2} + c_T \sup_{s \leq t} \big\lbrace  d_W\big(\hat{\mu}^n_{e,s} , \mu_{e,s} \big) + d_W\big(\hat{\mu}^n_{i,s} , \mu_{i,s} \big) \big\rbrace .
\end{align}
\end{lemma}
We finish with a proof of Theorem \ref{Big Theorem}.
\begin{proof}
An application of the Borel-Cantelli Lemma to Corollary \ref{Corollary Empirical Mean} implies that, with unit probability,
\begin{align}
\lim_{n\to\infty} n^{-1} \sup_{t\leq T} \sum_{j\in I_n} x_{e,t}^j &= 0 \\
\lim_{n\to\infty} n^{-1} \sup_{t\leq T} \sum_{j\in I_n} x_{i,t}^j &= 0 .
\end{align}
This means that
\begin{align}
\lim_{n\to\infty} \sup_{t\leq T} \big| \bar{v}_e(t) - v^n_e(t) \big| &= 0 \text{ and }
\lim_{n\to\infty} \sup_{t\leq T} \big| \bar{v}_i(t) - v^n_i(t) \big| &=0.
\end{align}
It thus suffices that we show that, with unit probability, for any $T < \eta$, 
\begin{align}
\lim_{n\to\infty} \sup_{t\leq T}  \big| K^n_e(t) - K_e(t) \big|   &= 0 \\
\lim_{n\to\infty} \sup_{t\leq T}  \big| K^n_i(t) - K_i(t) \big| &= 0 \\ 
\lim_{n\to\infty} \sup_{t\leq T}  \big| \bar{v}_e(t) - v^n_e(t) \big| &= 0 \\
\lim_{n\to\infty} \sup_{t\leq T}  \big| \bar{v}_i(t) - v^n_i(t) \big| &= 0.
\end{align}
\end{proof}
\subsection{Proofs}
We start by proving that the hydrodynamic limit is uniquely well-posed (i.e. we prove Lemma \ref{Hypothesis Manifold Evolution}).
\begin{proof}
Define $\big( \bar{v}_e , \bar{v}_i , K_e , K_i  \big)$ to satisfy the system of ODEs \eqref{eq: invariant covariance 1} , \eqref{eq: invariant covariance 2} and \eqref{eq: v ode} (below), i.e.
\begin{align}
\frac{d\vec{v}}{dt} =& -J_v^{-1} J_K\big(v_e(t) , v_i(t) ,  K_e(t), K_i(t)\big) \frac{d\vec{K}}{dt} \text{ where } \label{eq: v ode} \\
\frac{d\vec{K}}{dt} =& \left(\begin{array}{c}
\frac{dK_e}{dt} \\
\frac{dK_i}{dt} 
\end{array}\right).
\end{align}
Here the Jacobian of derivatives with respect to $K_e$ and $K_i$ is
\begin{align}
 J_K(\bar{v}_e , \bar{v}_i ,  K_e, K_i) = \left( \begin{array}{c c}
 \partial_{K_e} F_e(\bar{v}_e , \bar{v}_i ,  K_e, K_i) &  \partial_{K_i} F_e(\bar{v}_e , \bar{v}_i ,  K_e, K_i) \\
  \partial_{K_e} F_i(\bar{v}_e , \bar{v}_i ,  K_e, K_i) &  \partial_{K_i} F_i(\bar{v}_e , \bar{v}_i ,  K_e, K_i) 
 \end{array}\right).
\end{align}
The system of ODEs has a locally Lipschitz RHS. There thus exists a local solution. The Lipschitz coefficients are bounded over any set such that $\det(J_v) \geq \epsilon > 0$. Thus for any $\epsilon > 0$, there exists a unique solution up to the time
\begin{equation}
\eta = \lim_{\epsilon \to 0^+} \eta_{\epsilon}
\end{equation}
where
 \begin{align}
 \eta_{\epsilon} = \inf\big\lbrace t\geq 0 \; : \; \det(J_v) = \epsilon \big\rbrace,
 \end{align}
 and note that $\eta_{\epsilon}$ could be $\infty$.  The implicit function theorem implies that for all $t < \eta$, it must be that \eqref{eq: Fe zero equation} and \eqref{eq: Fi zero equation} are satisfied.

Finally we notice (after differentiating \eqref{eq: Fe zero equation} and \eqref{eq: Fe zero equation} ) that any solution satisfying the conditions in the Lemma must also be a solution of the ODEs  \eqref{eq: invariant covariance 1}, \eqref{eq: invariant covariance 2} and \eqref{eq: v ode}.
\end{proof}

We next prove Lemma \ref{Lemma Convergence to Gaussian}.
\begin{proof}
 We first outline a means of coupling the systems without interaction. Write $S = T C_f$, recalling that $C_f$ is the uniform upper bound for the spiking intensity.  We use classical work due to Komlos, Major, and Tusnady \cite{Komlos1975,Komlos1976}. These results imply that there exists a filtered probability space $\tilde{\mathcal{F}} = \big( \tilde{\mathcal{F}}_t \big)_{t \leq [0,S]}$ containing the variables $\big\lbrace Y^j_{\alpha\beta}(t) \big\rbrace_{j\in I_n \fatsemi \alpha,\beta \in \lbrace e,i \rbrace }$ (with the same distribution as above), and also Brownian Motions $\lbrace w^j_{\alpha\beta}(t) \rbrace_{j\in I_n}$ such that  (i) all variables are adapted to the filtration and are martingales, (ii) $\big\lbrace w^j_{\alpha\beta} \big\rbrace_{j\in I_n \fatsemi \alpha,\beta\in \lbrace e,i \rbrace}$ are independent Brownian Motions, (iii) writing
\begin{equation}
y^j_{\alpha\beta}(t) = n^{-1/2} Y^j_{\alpha\beta}(nt) - \sqrt{n}t
\end{equation}
it holds that for any $\epsilon > 0$, there exists $C_{\epsilon} > 0$ and $\mathfrak{n}_{\epsilon}$ such that for all $n \geq \mathfrak{n}_{\epsilon}$, 
\begin{align}
\mathbb{P}\big( \mathcal{X}_{\epsilon,n}^c \big) \leq \exp\big( - C_{\epsilon} n^{- 1/2} \big),
\end{align}
where
\begin{equation}
\mathcal{X}_{\epsilon,n} = \bigg\lbrace \sup_{j\in I_n} \sup_{t\leq S} \sup_{\alpha,\beta \in \lbrace e,i\rbrace} \big| y^j_{\alpha\beta}(t) - w^j_{\alpha\beta}(t) \big| \leq \epsilon   \bigg\rbrace .
\end{equation}
Hence $\mathcal{X}_{\epsilon,n}$ holds $\mathbb{P}$-almost-surely for all large enough $n$ (thanks to the Borel-Cantelli Lemma).

We now define
 \begin{align}
 \tilde{x}^j_{\alpha,t} =& u^j_{\alpha,0} - v^n_{\alpha}(0) - \tau_{\alpha}^{-1} \int_0^t \tilde{x}^j_{\alpha,s} ds + \tilde{W}^j_{\alpha e}(t) -  \tilde{W}^j_{\alpha i}(t) \\
  \grave{x}^j_{\alpha,t} =&  - \tau_{\alpha}^{-1} \int_0^t \grave{x}^j_{\alpha,s} ds + \grave{W}^j_{\alpha e}(t) -  \grave{W}^j_{\alpha i}(t)  \\
  \tilde{W}^j_{\alpha \beta}(t) =&  C_{\alpha \beta}  w^j_{\alpha \beta}\bigg( \int_0^t Q^n_{\alpha\beta}(s) ds \bigg)\\
 W^j_{\alpha\beta}(t) =&   C_{\alpha \beta}   y^j_{\alpha \beta}\bigg( \int_0^t Q^n_{\alpha\beta}(s) ds \bigg)
 \end{align}
 As long as the event $\mathcal{X}_n$ holds, it must be that 
 \begin{align}
 \sup_{t\leq T} \big|  W^j_{\alpha\beta}(t) -  \tilde{W}^j_{\alpha\beta}(t) \big| \leq C_{max} \epsilon,
 \end{align}
 where $C_{max} = \sup_{\alpha,\beta \in \lbrace e,i \rbrace} |C_{\alpha\beta}|$.
 An application of Gronwall's Inequality then implies that there exists a constant $c > 0$ such that for all $n\geq \mathfrak{n}_\epsilon$, (and as long as the event $\mathcal{X}_{\epsilon,n}$ holds),
 \begin{align}
\sup_{ j\in I_n} \sup_{t\leq T}\big| x^{j}_{\alpha,t} - \tilde{x}^j_{\alpha,t} \big| \leq c \epsilon. \label{eq: bound the x difference}
 \end{align}
Write
\begin{align}
\tilde{u}^j_{\alpha,t} = \tilde{x}^j_{\alpha,t} + v_{\alpha,t}.
\end{align} 
Define $\tilde{\mu}^n_{\alpha,t} \in \mathcal{P}\big( \mathcal{C}([0,T] , \mathbb{R} \big)$ to be the empirical measure
\begin{align}
\tilde{\mu}^n_{\alpha} = n^{-1}\sum_{j\in I_n} \delta_{\tilde{x}^j_{\alpha}}.
\end{align}
It is immediate from \eqref{eq: bound the x difference} that (as long as the event $\mathcal{X}_{\epsilon,n}$ holds),
\begin{align}
d_W\big( \tilde{\mu}^n_{\alpha} , \hat{\mu}^n_{\alpha} \big) \leq c \epsilon.
\end{align}
Thus in order for the lemma to hold, it suffices that we prove Lemma \ref{Lemma Intermediate Bound}.
\end{proof}
\begin{lemma} \label{Lemma Intermediate Bound}
For any $\epsilon > 0$,
\begin{equation}
\lsup{n} n^{a-1} \log \mathbb{P}\bigg( \text{For some }\alpha\in \lbrace e,i\rbrace , \; \; d_W\big( \tilde{\mu}^n_{\alpha} , \nu_{\alpha} \big) \geq \epsilon  \bigg) < 0
\end{equation}
\end{lemma}
\begin{proof}

 Let $\Upsilon \subset \mathcal{C}([0,T], \mathbb{R}\big)$ consist of all functions $H$ that are (i) Lipschitz, with Lipschitz constant less than or equal to $C_f$, and (ii) such that $H(0) = 0$, and (iii) if $t \geq s$ then necessarily $H(t) \geq H(s)$. Since $\Upsilon$ is compact (thanks to the Arzela-Ascoli Theorem), for any $\epsilon > 0$ it admits a covering of the form
 \begin{align}
 \Upsilon \subseteq \bigcup_{i=1}^{M_{\delta}} B_{\delta}\big(g_{\delta,i}\big) 
 \end{align}
 for some $ \lbrace g_{\delta,i} \rbrace_{1\leq i \leq M_{\delta}}\subset \Upsilon$, and writing
 \[
 B_{\delta}(g_{\delta,i}) = \big\lbrace f \in \Upsilon \; : \; \sup_{t\leq T} \big| f(t) - g_{\delta,i}(t) \big| < \delta \big\rbrace.
 \]
 One checks that the functions can be chosen such that $|M_{\delta}| \leq \rm{Const} \times \delta^{-2}$, for a constant that holds for all $\delta \leq 1$.   Write $\mathcal{V}_{\delta} = \big\lbrace g_{\delta,i} \big\rbrace_{1\leq i \leq M_{\delta}}$. Employing a union-of-events bound, we now find that
\begin{multline}
 \mathbb{P}\bigg( \text{For some }\alpha \in \lbrace e,i\rbrace , \; \; d_W\big( \tilde{\mu}^n_{\alpha} , \nu_{\alpha} \big) \geq \epsilon  \bigg) \\
 \leq \sum_{\alpha\in \lbrace e,i\rbrace  } \sum_{k,l=1}^{M_{\delta}} \mathbb{P}\bigg(  d_W\big( \tilde{\mu}^n_{\alpha} , \nu_{\alpha} \big) \geq \epsilon , \Lambda^n_{\alpha e} \in B_{\delta}(g_{\delta,k}) ,  \Lambda^n_{\alpha i} \in B_{\delta}(g_{\delta,l}) \bigg). \label{eq: decompose probability summation}
\end{multline}
Fixing $g_{\delta,k}$ and $g_{\delta,l}$, write
\begin{align}
\breve{W}^j_{\alpha\beta,\delta,k}(t) =& C_{\alpha\beta} w^j_{\alpha \beta}\big( g_{\delta,k}(t) \big) \\
 \breve{x}^j_{\alpha,\delta,k,l,t} =& u^j_{\alpha,0} - v^n_{\alpha}(0) - \tau_{\alpha}^{-1} \int_0^t \breve{x}^j_{\alpha,\delta,k,l,s} ds + \breve{W}^j_{\alpha e,\delta , k}(t) -  \breve{W}^j_{\alpha i , \delta , l}(t).
\end{align}
Writing
\begin{align}
\breve{\mu}^n_{\alpha,\delta,k,l} = n^{-1}\sum_{j\in I_n} \delta_{\breve{x}^j_{\alpha,\delta,k,l}} \in \mathcal{P}\big( \mathcal{C}([0,T],\mathbb{R}) \big),
\end{align}
the triangle inequality implies that
\begin{align}
d_W\big( \tilde{\mu}^n_{\alpha} , \nu_{\alpha} \big) \leq d_W\big( \tilde{\mu}^n_{\alpha} ,\breve{\mu}^n_{\alpha,\delta,k,l} \big) + d_W\big(\breve{\mu}^n_{\alpha,\delta,k,l}, \nu_{\alpha} \big) .
\end{align}
In light of \eqref{eq: decompose probability summation}, it thus suffices that we prove that, as long as $\delta$ is sufficiently small,
\begin{align}
\sup_{\alpha\in \lbrace e,i\rbrace  } \sup_{1\leq k,l \leq M_{\delta}}  \lsup{n} n^{-1}\log \mathbb{P}\bigg(  d_W\big( \tilde{\mu}^n_{\alpha} , \breve{\mu}^n_{\alpha,\delta,k,l} \big) \geq \epsilon / 2 , \Lambda^n_{\alpha e} \in B_{\delta}(g_{\delta,k}) ,  \Lambda^n_{\alpha i} \in B_{\delta}(g_{\delta,l}) \bigg) < 0. \label{eq: to prove gaussian approx 1} \\
\sup_{\alpha\in \lbrace e,i\rbrace  } \sup_{1\leq k,l \leq M_{\delta}}  \lsup{n} n^{-1}\log \mathbb{P}\bigg(  d_W\big( \nu_{\alpha} , \breve{\mu}^n_{\alpha,\delta,k,l} \big) \geq  \epsilon / 2 , \Lambda^n_{\alpha e} \in B_{\delta}(g_{\delta,k}) ,  \Lambda^n_{\alpha i} \in B_{\delta}(g_{\delta,l}) \bigg) < 0.\label{eq: to prove gaussian approx 2}
\end{align}
Starting with \eqref{eq: to prove gaussian approx 1}, one easily checks that (analogously to above )
\begin{align}
\sup_{t\leq T} \big|  \breve{x}^j_{\alpha,k,l,t} - \tilde{x}^j_{\alpha,t} \big| \leq 2 \sup_{t\leq T} \big| \breve{W}^j_{\alpha e,\delta , k}(t) - \tilde{W}^j_{\alpha e}(t)\big| + 2\sup_{t\leq T} \big|  \breve{W}^j_{\alpha i , \delta , l}(t) -  \tilde{W}^j_{\alpha i}(t) \big| . \label{eq: intermediate breve x tilde x bound}
\end{align}
It follows from \eqref{eq: intermediate breve x tilde x bound} that
\begin{align}
d_W\big( \tilde{\mu}^n_{\alpha} ,\breve{\mu}^n_{\alpha,\delta,k,l} \big) \leq \frac{2}{n}\sum_{j\in I_n}\bigg\lbrace \sup_{t\leq T} \big| \breve{W}^j_{\alpha e,\delta , k}(t) - \tilde{W}^j_{\alpha e}(t)\big| + \sup_{t\leq T} \big|  \breve{W}^j_{\alpha i , \delta , l}(t) -  \tilde{W}^j_{\alpha i}(t) \big| \bigg\rbrace .
\end{align}
If $ \Lambda^n_{\alpha e} \in B_{\delta}(g_{\delta,k}) $, then necessarily $\sup_{t\leq T} \big| \Lambda^n_{\alpha w}(t) - g_{\delta,k}(t) \big| \leq \delta$. This in turn implies that
\begin{align}
\sup_{t\leq T} \big| \breve{W}^j_{\alpha e,\delta , k}(t) - \tilde{W}^j_{\alpha e}(t)\big| \leq \sup_{s,t\leq S : |s-t | \leq C_f \delta} \big| w^j_{\alpha e}(t) - w^j_{\alpha e}(s) \big| 
\end{align}
We analogously find that if $ \Lambda^n_{\alpha i} \in B_{\delta}(g_{\delta,l}) $, then necessarily $\sup_{t\leq T} \big| \Lambda^n_{\alpha i}(t) - g_{\delta,l}(t) \big| \leq \delta$. This in turn implies that
\begin{align}
\sup_{t\leq T} \big| \breve{W}^j_{\alpha i,\delta , l}(t) - \tilde{W}^j_{\alpha i}(t)\big| \leq \sup_{s,t\leq S : |s-t | \leq C_f \delta} \big| w^j_{\alpha i}(t) - w^j_{\alpha i}(s) \big| .
\end{align}
In summary, we find that if $ \Lambda^n_{\alpha e} \in B_{\delta}(g_{\delta,k}) $ and $ \Lambda^n_{\alpha i} \in B_{\delta}(g_{\delta,l}) $, then necessarily
\begin{align}
d_W\big( \tilde{\mu}^n_{\alpha} ,\breve{\mu}^n_{\alpha,\delta,k,l} \big) \leq \frac{2}{n} \sum_{j\in I_n} \bigg\lbrace  \sup_{s,t\leq S : |s-t | \leq C_f \delta} \big| w^j_{\alpha e}(t) - w^j_{\alpha e}(s) \big|  +  \sup_{s,t\leq S : |s-t | \leq C_f \delta} \big| w^j_{\alpha i}(t) - w^j_{\alpha i}(s) \big| \bigg\rbrace .
\end{align}
We thus find that, by taking $\delta$ to be sufficiently small,
\begin{multline}
 \lsup{n} n^{-1}\log \mathbb{P}\bigg(  d_W\big( \tilde{\mu}^n_{\alpha} , \breve{\mu}^n_{\alpha,\delta,k,l} \big) \geq \epsilon / 2, \Lambda^n_{\alpha e} \in B_{\delta}(g_{\delta,k}) ,  \Lambda^n_{\alpha i} \in B_{\delta}(g_{\delta,l}) \bigg) \\
 \leq  \lsup{n} n^{-1}\log \bigg\lbrace M_{\delta}^2 \mathbb{P}\bigg(  \frac{2}{n} \sum_{j\in I_n} \bigg\lbrace  \sup_{s,t\leq S : |s-t | \leq C_f \delta} \big| w^j_{\alpha e}(t) - w^j_{\alpha e}(s) \big|\\  +  \sup_{s,t\leq S : |s-t | \leq C_f \delta} \big| w^j_{\alpha i}(t) - w^j_{\alpha i}(s) \big| \bigg\rbrace \geq \epsilon / 4 \bigg) \bigg\rbrace  < 0
\end{multline}
where the last step holds as long as $\delta$ is sufficiently small, thanks to Lemma \ref{Lemma Appendix} in the Appendix.

It remains to prove \eqref{eq: to prove gaussian approx 2}. To this end, define $\nu_{\alpha,\delta,k,l} \in \mathcal{P}\big( \mathcal{C}([0,T],\mathbb{R}) \big)$ to be the law of $\breve{x}^1_{\alpha,k,l}$. Thanks to the triangle inequality,
 \begin{align}
 d_W\big( \nu_{\alpha} , \breve{\mu}^n_{\alpha,\delta,k,l} \big) \leq d_W\big( \nu_{\alpha} , \nu_{\alpha,\delta,k,l} \big) + d_W\big( \nu_{\alpha,\delta,k,l} ,  \breve{\mu}^n_{\alpha,\delta,k,l} \big).
 \end{align}
In order that \eqref{eq: to prove gaussian approx 2} holds, it thus suffices that we prove that for arbitrary $\tilde{\epsilon} > 0$, as long as $\delta$ is small enough,
\begin{align}
\sup_{\alpha\in \lbrace e,i\rbrace  } \sup_{1\leq k,l \leq M_{\delta}}   \bigg\lbrace d_W\big( \nu_{\alpha} , \nu_{\alpha,\delta,k,l} \big) \chi\big\lbrace \Lambda^n_{\alpha e} \in B_{\delta}(g_{\delta,k}) ,  \Lambda^n_{\alpha i} \in B_{\delta}(g_{\delta,l}) \big\rbrace \bigg\rbrace \leq .\label{eq: to prove gaussian approx 3} \\
\sup_{\alpha\in \lbrace e,i\rbrace  }  \lsup{n} \sup_{1\leq k,l \leq M_{\delta}} n^{-1}\log \mathbb{P}\bigg(  d_W\big( \nu_{\alpha,\delta,k,l} ,  \breve{\mu}^n_{\alpha,\delta,k,l} \big) \geq \tilde{\epsilon} , \Lambda^n_{\alpha e} \in B_{\delta}(g_{\delta,k}) ,  \Lambda^n_{\alpha i} \in B_{\delta}(g_{\delta,l}) \bigg) < 0.\label{eq: to prove gaussian approx 4}
\end{align}
Now \eqref{eq: to prove gaussian approx 3} follows from the definitions. This is because $\nu_{\alpha}$ and $\nu_{\alpha,\delta,k,l}$ are both centered and Gaussian. The covariances must also converge uniformly as $\delta \to 0$, as long as $\Lambda^n_{\alpha e} \in B_{\delta}(g_{\delta,k})$ and $\Lambda^n_{\alpha i} \in B_{\delta}(g_{\delta,l})$.

\eqref{eq: to prove gaussian approx 4} is a consequence of Sanov's Theorem \cite{Dembo1998}. This is because $\lbrace \breve{x}^j_{\alpha,k,l} \rbrace_{j\in I_n}$ are independent and identically distributed, each with probability law $\nu_{\alpha,k,l}$. 
\end{proof}
We next prove Lemma \ref{Lemma bound v n e t v n i t}.
\begin{proof}
Define
\begin{align}
\zeta^n_s = \sup_{s \leq t} \big\lbrace | v^n_{e}(s) -  \bar{v}_e(s) |
+ | v^n_{i}(s) -  \bar{v}_{i}(s) | + L^n_s \big\rbrace 
\end{align}
It follows from Lemmas \ref{Lemma intermedaite L n t bound} and \ref{Lemma bound v norm difference} that there exists a constant $a_T > 0$ such that for all $t\leq T$,
\[
\zeta_t \leq a_T \int_0^t \zeta_s ds + a_T n^{-1/2}+ a_T  d_W\big( \hat{\mu}^n_{e,t} , \nu_{e,t} \big)  + a_T  d_W\big( \hat{\mu}^n_{i,t} , \nu_{i,t} \big)
\]
We can thus apply Gronwall's Inequality to obtain the Lemma.
\end{proof}
\begin{lemma} \label{Lemma intermedaite L n t bound}
There exists a constant $C_T > 0$ such that for all $t\leq T$ and all $n\geq 1$,
\begin{align}
L^n_t \leq C_T t \sup_{s \leq t}\big( \big| v^n_e(s) - \bar{v}_e(s) \big| +\big| v^n_i(s) - \bar{v}_i(s) \big| \big) +   t d_W\big( \hat{\mu}^n_{e,t} , \nu_{e,t} \big)  +t d_W\big( \hat{\mu}^n_{i,t} , \nu_{i,t} \big)
\end{align}
\end{lemma}
\begin{proof}
Define
\begin{align}
\tilde{L}^n_t = \sup_{s \leq t} \big| \mathbb{E}^{\nu^n_e}[  x_s^2 ] - \mathbb{E}^{\mu_e}[x_s^2] \big| + \sup_{s \leq t} \big| \mathbb{E}^{\nu^n_i}[ x_s^2 ] - \mathbb{E}^{\mu_i}[ x_s^2] \big|.
\end{align}
We first claim that there exists a constant $C > 0$ such that for all $t\leq T$,
\begin{align} \label{eq: tilde L n t claim}
\tilde{L}^n_t \leq C t \sup_{s\leq t}  \sup_{\alpha,\beta \in \lbrace e,i \rbrace } \bigg\lbrace \bigg| Q^n_{\alpha\beta}(s) -
C^2_{\alpha\beta } \int_{\mathbb{R}} \rho(K_{\alpha\beta}(s),x) f_{\alpha \beta} ( \bar{v}_{\beta}(s) + x) dx \bigg| \bigg\rbrace .
\end{align}
Indeed standard identities for the Ornstein-Uhlenbeck process dictate that for $\alpha \in \lbrace e,i \rbrace$,
\begin{align}
\frac{d}{dt}\mathbb{E}^{\nu^n_\alpha}[ x_t^2 ] = -\frac{2}{\tau_{\alpha}}\mathbb{E}^{\nu^n_\alpha}[ x_t^2 ] + C_{\alpha e}^2 Q^n_{\alpha e}(t) +  C_{\alpha i}^2 Q^n_{\alpha i}(t). \label{eq: comparison 1}
\end{align}
We similarly find that
\begin{multline}
\frac{d}{dt}\mathbb{E}^{\mu_\alpha}[ x_t^2 ] = -\frac{2}{\tau_{\alpha}}\mathbb{E}^{\mu_\alpha}[ x_t^2 ] + C_{\alpha e}^2  \int_{\mathbb{R}} \rho(K_{\alpha e}(s),x) f_{\alpha e} ( \bar{v}_{e}(s) + x) dx\\ +  C_{\alpha i}^2\int_{\mathbb{R}} \rho(K_{\alpha i}(s),x) f_{\alpha i} ( \bar{v}_{i}(s) + x) dx .\label{eq: comparison 2}
\end{multline}
\eqref{eq: tilde L n t claim} now follows from an application of Gronwall's Inequality to \eqref{eq: comparison 1} and \eqref{eq: comparison 2}.

For any $s > t$, standard results on the Ornstein-Uhlenbeck process imply that for each $\alpha \in \lbrace e,i\rbrace$,
\begin{align}
\frac{d}{ds}\mathbb{E}^{\nu^n_\alpha}[ x_t x_s ] =& - \tau_{\alpha}^{-1}\mathbb{E}^{\nu^n_\alpha}[ x_t x_s ] \label{eq: tilde L n claim 3} \\
\frac{d}{ds}\mathbb{E}^{\mu_\alpha}[ x_t x_s ] =& - \tau_{\alpha}^{-1} \mathbb{E}^{\mu_\alpha}[ x_t x_s ] .\label{eq: tilde L n claim 4} 
\end{align}
\eqref{eq: tilde L n claim 3}  and \eqref{eq: tilde L n claim 4} imply that
\begin{align}
\tilde{L}^n_t = L^n_t,
\end{align}
and therefore \eqref{eq: tilde L n t claim} implies that
\begin{align} \label{eq: L n t temporary inequality}
L^n_t \leq C t \sup_{s\leq t} \sup_{\alpha,\beta \in \lbrace e,i \rbrace }\bigg\lbrace \bigg| Q^n_{\alpha\beta}(s) -
C^2_{\alpha\beta } \int_{\mathbb{R}} \rho(K_{\alpha\beta}(s),x) f_{\alpha \beta} ( \bar{v}_{\beta}(s) + x) dx \bigg| \bigg\rbrace .
\end{align}
Substituting definitions, it holds that there exists a constant $c > 0$ such that for all $\alpha,\beta \in \lbrace e,i \rbrace$,
\begin{multline}
\bigg| Q^n_{\alpha \beta}(t) -  \int_{\mathbb{R}} \rho(K_{\alpha\beta}(t),x) C^2_{\alpha\beta } f_{\alpha \beta} ( \bar{v}_{\beta}(t) + x) dx \bigg| \leq
\bigg| Q^n_{\alpha \beta}(t) -  \int_{\mathbb{R}} \rho(K^n_{\alpha\beta}(t),x) C^2_{\alpha\beta } f_{\alpha \beta} ( \bar{v}_{\beta}(t) + x) dx \bigg| \\+  C^2_{\alpha\beta }\bigg| \int_{\mathbb{R}}\big\lbrace \rho(K^n_{\alpha\beta}(t),x) - \rho(K_{\alpha\beta}(t),x) \big\rbrace  f_{\alpha \beta} ( \bar{v}_{\beta}(t) + x) dx \bigg|
\\ \leq c \big| \bar{v}_{\beta}(t) - v^n_{\beta}(t) \big| + c d_W\big( \hat{\mu}^n_{\beta,t} , \nu^n_{\beta,t} \big) + c \big| K^n_{\alpha\beta}(t) - K_{\alpha\beta}(t) \big|. \label{eq: bound Q n t difference}
\end{multline}
The Lemma now follows from \eqref{eq: L n t temporary inequality} and \eqref{eq: bound Q n t difference}.
\end{proof}

We next prove Lemma \ref{Lemma bound v norm difference}.
\begin{proof}
We compute that
\begin{multline}
\frac{d}{dt}\big( q^n_e(t)^2 + q^n_i(t)^2 \big) =   q^n_e(t) \bigg( -\tau_e^{-1}v^n_e(t)+ \frac{d\bar{v}_e(t)}{dt}    \bigg) + q^n_i(t) \bigg( -\tau_i^{-1}v^n_i(t) + \frac{d\bar{v}_i(t)}{dt} \bigg) \\+ \sqrt{n} \big( q^n_{e}(t) F^n_{e,t}-  q^n_{i}(t) F^n_{i,t} \big) 
\end{multline}
It follows from Lemma \ref{Hypothesis Manifold Evolution} that there must exist a constant $\tilde{c}_T$ such that for all $t\leq T$,
\begin{align*}
\bigg| -\tau_e^{-1}v^n_e(t) + \frac{d\bar{v}_e(t)}{dt}  \bigg| &\leq \tilde{c}_T \\
\bigg| -\tau_i^{-1}v^n_i(t) + \frac{d\bar{v}_i(t)}{dt}   \bigg| &\leq \tilde{c}_T.
\end{align*}
We can now employ Lemma \ref{Lemma Balanced Manifold} to conclude the Lemma.
\end{proof}
 
\begin{lemma} \label{Lemma Balanced Manifold}
There exist constants $c_T , \epsilon_T  > 0$  such that for all $n \geq 1$, if  
\begin{equation}
q^n_{e}(t)^2 + q^n_i(t)^2 \leq \epsilon_T
\end{equation}
then necessarily
\begin{multline}
q^n_{e}(t) F^n_{e,t}-  q^n_{i}(t) F^n_{i,t} \leq - \frac{\zeta_t}{2}  \big( q^n_{e}(t)^2 + q^n_i(t)^2 \big) + c_T  \big\lbrace d_W\big(\hat{\mu}^n_{e,t} , \mu_{e,t} \big) \\+ d_W\big(\hat{\mu}^n_{i,t} , \mu_{i,t} \big) \big\rbrace \big( q^n_{e}(t)^2 + q^n_i(t)^2 \big)^{1/2}. \label{eq: lip decay delta n}
\end{multline}
where
\begin{align}
 F^n_{e,t} := n^{-1}\sum_{j\in I_n} \bigg( C_{ee} f_{ee}(u^j_{e,t}) - C_{ei} f_{ei}(u^j_{i,t}) \bigg) \\
 F^n_{i,t} := n^{-1}\sum_{j\in I_n} \bigg( C_{ie} f_{ie}(u^j_{e,t}) - C_{ii} f_{ii}(u^j_{i,t}) \bigg) 
\end{align}
\end{lemma} 
\begin{proof}
Define the functions
\begin{align}
\hat{F}_e, \hat{F}_i  :& \mathbb{R}^2 \times \mathcal{P}(\mathbb{R}) \times  \mathcal{P}(\mathbb{R})   \mapsto \mathbb{R} \\
\hat{F}_e (v_{e} , v_i,  \gamma_e ,  \gamma_i ) =& C_{ee}  \int_{\mathbb{R}} f_{ee} ( v_e + x ) d\gamma_e(x) - C_{ei}  \int_{\mathbb{R}}   f_{ei}( v_i + x) d\gamma_i(x) \\
\hat{F}_i (v_{e} , v_i,  \gamma_e ,  \gamma_i ) =& C_{ie}  \int_{\mathbb{R}} f_{ie} ( v_e + x ) d\gamma_e(x) - C_{ii}  \int_{\mathbb{R}}   f_{ii}( v_i + x) d\gamma_i(x)   .
\end{align}
and we make the decomposition
\begin{multline}
q^n_{e}(t) F^n_{e,t}-  q^n_{i}(t) F^n_{i,t} = q^n_{e}(t) \big(\hat{F}_e (v^n_{e,t} , v^n_{i,t},  \mu_{e,t} ,  \mu_{i,t} ) - \hat{F}_e ( \bar{v}_{e,t} , \bar{v}_{i,t},  \mu_{e,t} ,  \mu_{i,t} ) \big)  \\
+ q^n_{i}(t) \big(\hat{F}_i (v^n_{e,t} , v^n_{i,t},  \mu_{e,t} ,  \mu_{i,t} ) - \hat{F}_i ( \bar{v}_{e,t} , \bar{v}_{i,t},  \mu_{e,t} ,  \mu_{i,t} ) \big) \\
+ q^n_e(t) \big(  \hat{F}_e (v^n_{e,t} , v^n_{i,t},  \hat{\mu}^n_{e,t} ,  \hat{\mu}^n_{i,t} ) - \hat{F}_e (v^n_{e,t} , v^n_{i,t},  \mu_{e,t} ,  \mu_{i,t} ) \big)\\
+ q^n_i(t) \big(  \hat{F}_i (v^n_{e,t} , v^n_{i,t},  \hat{\mu}^n_{e,t} ,  \hat{\mu}^n_{i,t} ) - \hat{F}_i (v^n_{e,t} , v^n_{i,t},  \mu_{e,t} ,  \mu_{i,t} ) \big), \label{eq: expand inner product q n F}
\end{multline}
noting that (by definition of the dynamics in Lemma \ref{Hypothesis Manifold Evolution}) 
\[
 \hat{F}_e ( \bar{v}_{e} , \bar{v}_i,  \mu_e ,  \mu_i ) =  \hat{F}_i ( \bar{v}_{e} , \bar{v}_i,  \mu_e ,  \mu_i )  = 0.
 \] 
Let the Jacobian $J: \mathbb{R}^2 \times \mathcal{P}(\mathbb{R})^2 \mapsto \mathbb{R}^{2 \times 2}$ be such that
\begin{align}
J(v_e,v_i,\mu_e , \mu_i) = \left(\begin{array}{c c}
\partial_{v_e} \hat{F}_e (v_{e} , v_i,  \mu_e ,  \mu_i ) & \partial_{v_i} \hat{F}_e (v_{e} , v_i,  \mu_e ,  \mu_i ) \\
\partial_{v_e} \hat{F}_i (v_{e} , v_i,  \mu_e ,  \mu_i ) & \partial_{v_i} \hat{F}_i (v_{e} , v_i,  \mu_e ,  \mu_i ) ,
\end{array}\right).
\end{align}
The Intermediate Value Theorem implies that
 \begin{multline}
 q^n_{e}(t) \big(\hat{F}_e (v^n_{e} , v^n_i,  \mu_e ,  \mu_i ) - \hat{F}_e ( \bar{v}_{e} , \bar{v}_i,  \mu_e ,  \mu_i ) \big)  \\
+ q^n_{i}(t) \big(\hat{F}_i (v^n_{e} , v^n_i,  \mu_e ,  \mu_i ) - \hat{F}_i ( \bar{v}_{e} , \bar{v}_i,  \mu_e ,  \mu_i ) \big) =  \vec{q}(t)^T J( \tilde{v}_e, \tilde{v}_i,\mu_e , \mu_i) \vec{q}(t)
\end{multline}
where for some $a\in [0,1]$, $\tilde{v}_e = a \bar{v}_e(t) + (1-a) v^n_e(t)$ and  $\tilde{v}_i = a \bar{v}_i(t) + (1-a) v^n_i(t)$. 

Now, since the balanced manifold is by definition attracting, for the constant $\zeta_t$ defined just after Lemma \ref{Hypothesis Manifold Evolution}, 
\[
\vec{q}(t)^T J\big( \bar{v}_{e,t}, \bar{v}_{i,t} ,\mu_{e,t} , \mu_{i,t} \big) \vec{q}(t) \leq - \zeta_t \big( q_e(t)^2 + q_i(t)^2 \big).
\]
Now the map $(v_e,v_i) \mapsto J( v_e, v_i,\mu_{e,t}, \mu_{i,t} )$ is differentiable, which means that as long as $q_e(t)^2 + q_i(t)^2$ is sufficiently small it must be that the real parts of the eigenvalues of 
$J( \tilde{v}_e, \tilde{v}_i,\mu_{e,t} , \mu_{i,t})$ are less than or equal to $-\zeta_t / 2$. We thus find that as long as $q_e(t)^2 + q_i(t)^2$ is sufficiently small (which is ensured by taking $\epsilon_T$ sufficiently small), it must be that 
\begin{equation}
\vec{q}(t)^T J\big( \tilde{v}_{e,t}, \tilde{v}_{i,t} ,\mu_{e,t} , \mu_{i,t} \big) \vec{q}(t) \leq - \frac{\zeta_t}{2} \big( q_e(t)^2 + q_i(t)^2 \big).
\end{equation}
For the other terms in \eqref{eq: expand inner product q n F}, the definition of the Wasserstein Distance implies that there exists a constant $c$ such that 
\begin{align}
\big|  \hat{F}_e (v^n_{e,t} , v^n_{i,t},  \hat{\mu}^n_{e,t} ,  \hat{\mu}^n_{i,t} ) - \hat{F}_e (v^n_{e,t} , v^n_{i,t},  \mu_{e,t} ,  \mu_{i,t} ) \big| &\leq c \big( d_W\big( \hat{\mu}^n_{e,t} , \mu_{e,t} \big) +  d_W\big( \hat{\mu}^n_{i,t} , \mu_{i,t} \big) \big) \\
\big| \hat{F}_i (v^n_{e,t} , v^n_{i,t},  \hat{\mu}^n_{e,t} ,  \hat{\mu}^n_{i,t} ) - \hat{F}_i (v^n_{e,t} , v^n_{i,t},  \mu_{e,t} ,  \mu_{i,t} ) \big| &\leq c  \big( d_W\big( \hat{\mu}^n_{e,t} , \mu_{e,t} \big) +  d_W\big( \hat{\mu}^n_{i,t} , \mu_{i,t} \big) \big) .
\end{align}
We can now conclude that the Lemma holds.
%
\end{proof}


 
\section{Numerical Simulations}

In this section we compare the analytic results to stochastic simulations. The firing-rate functions are assumed to be of the form
\begin{align}
f_{ee}(x) =& 0.5* \big( \tanh( c_{ee} x) + 2 \big) \\
f_{ei}(x) =& 0.5* \big( \tanh( c_{ie} x) + 2 \big)  \\
f_{ii}(x) =&   \tanh(c_{ii} x) + 1 \\
f_{ie}(x) =&   \tanh(c_{ei} x) + 1 
\end{align}
where $\lbrace c_{\alpha\beta} \rbrace_{\alpha,\beta \in \lbrace e,i \rbrace}$ are positive constants. It is essential that our initial conditions lie on the balanced manifold (as noted in Hypothesis \ref{Hypothesis Initial Conditions on the Balanced Manifold}). We fix the desired empirical covariances at time $0$ (written as $k_e > 0$ and $k_i > 0$). Let $m_e$ and $m_i$ be (respectively) the mean excitatory activity at time $0$, and the mean inhibitory activity at time $0$. We solve for $m_e$ and $m_i$ numerically by requiring that $F_e(m_e,m_i,k_e,k_i) = 0$ and $F_i(m_e,m_i,k_e,k_i) = 0$. Solutions for $m_e$ and $m_i$ were obtained using MATLAB's \texttt{fsolve} function. For the following examples, it was verified that the initial condition lies on the balanced manifold (i.e. the linearized dynamics is stable in a neighborhood of the initial condition).

For the stochastic simulations, the initial conditions were chosen to be such that
\begin{align}
u^j_{e,0} &= m_e + \sqrt{k_e} \tilde{x}^j_{e,0} \\
u^j_{i,0} &= m_i + \sqrt{k_i} \tilde{x}^j_{i,0} .
\end{align}
where $\lbrace \tilde{x}^j_{e,0} , \tilde{x}^j_{i,0} \rbrace_{j\in I_n}$ are independent samples from $\mathcal{N}(0,1)$.

\medskip
In what follows, we use the parameter values:
$C_{ee} = 1, \, C_{ei} = 1.5, \, C_{ie} = 0.5, \, C_{ii} = 0.5, \, c_{ee} = 1, \, c_{ii} = 0.5, \, \tau_e = 1, \, \tau_i = 1, \, n = 10000, \, n_e = 5000, \, n_i = 5000 \,.$

\subsection{Results for $k_e=k_i=1$}





    


\begin{center}
    \begin{minipage}{0.49\textwidth}
        \centering
        \includegraphics[width=\textwidth]{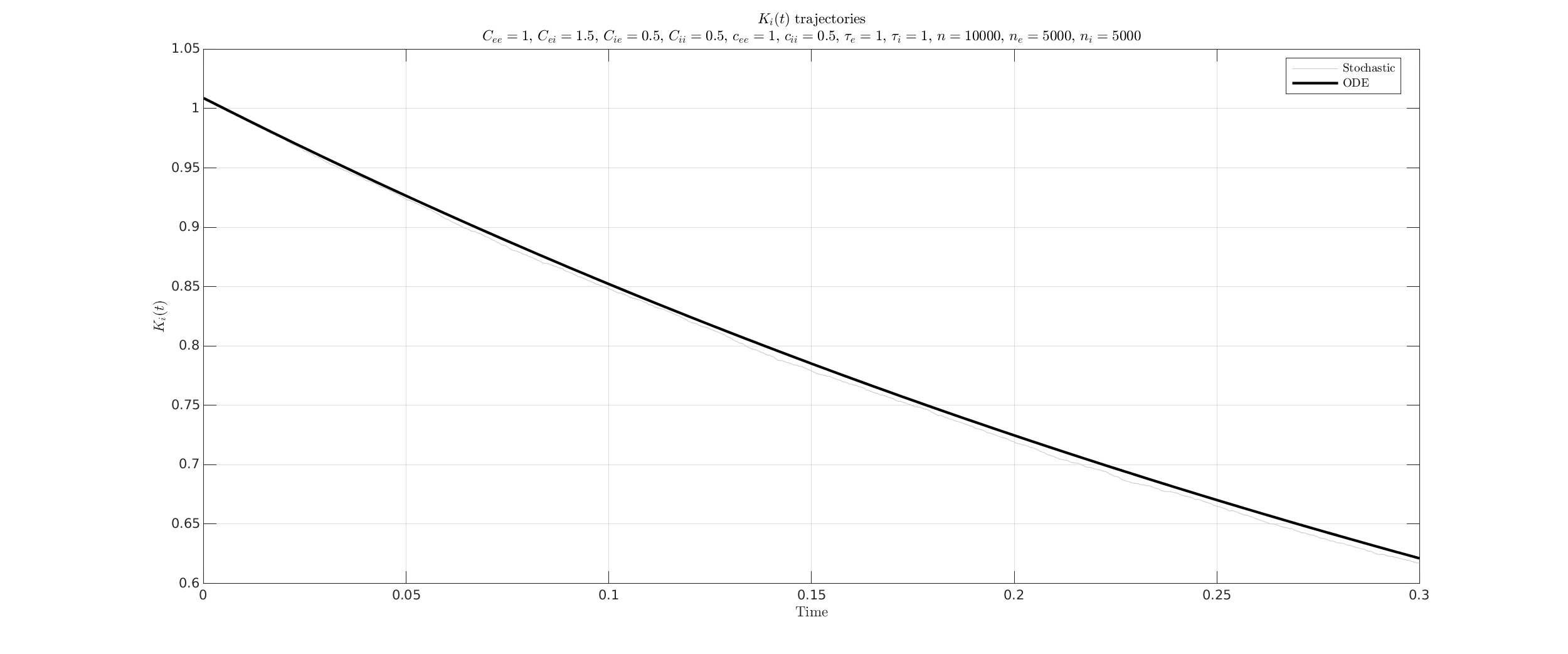}
        {\small $K^n_i(t)$ vs. $K_i(t)$}
        \label{fig:comparison_K_i}
    \end{minipage}
    \hfill
    \begin{minipage}{0.49\textwidth}
        \centering
        \includegraphics[width=\textwidth]{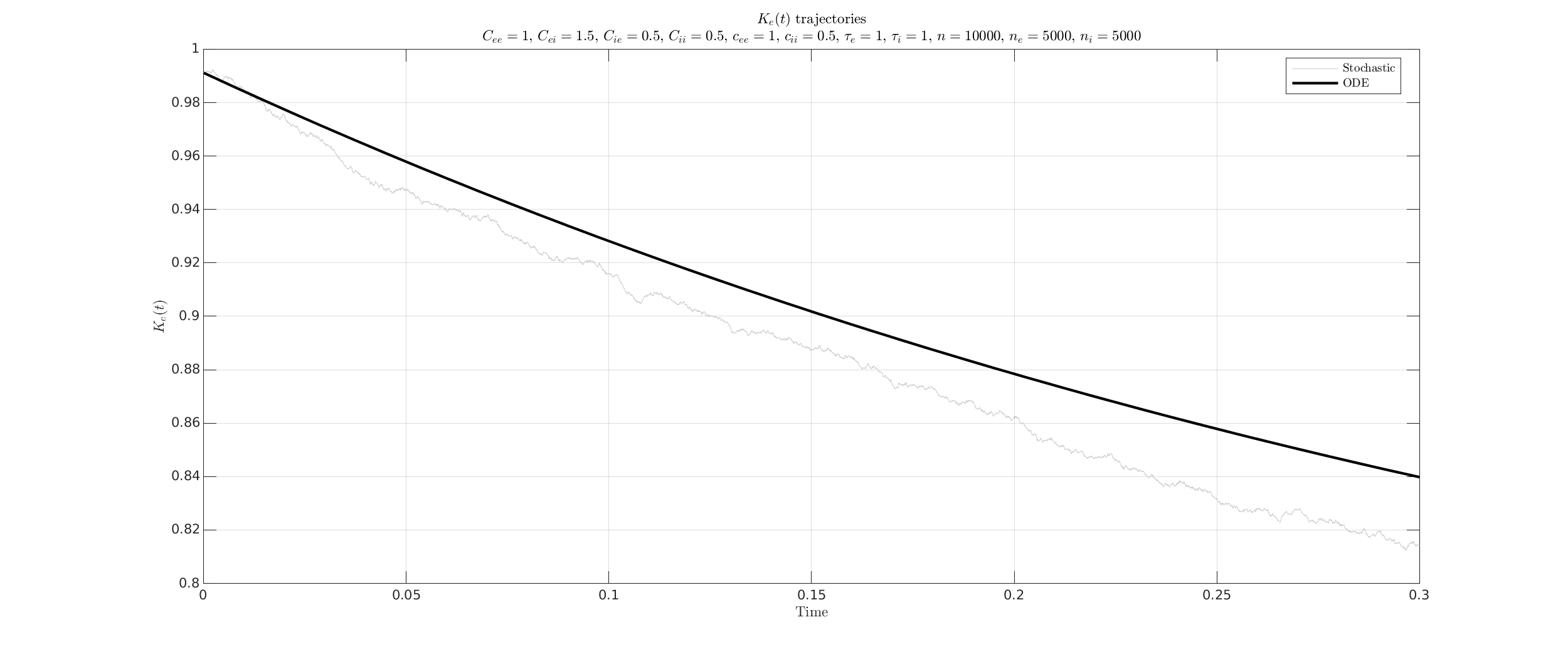}
        {\small $K^n_e(t)$ vs. $K_e(t)$}
        \label{fig:comparison_K_e}
    \end{minipage}

    \vspace{0.5cm}

    \begin{minipage}{0.49\textwidth}
        \centering
        \includegraphics[width=\textwidth]{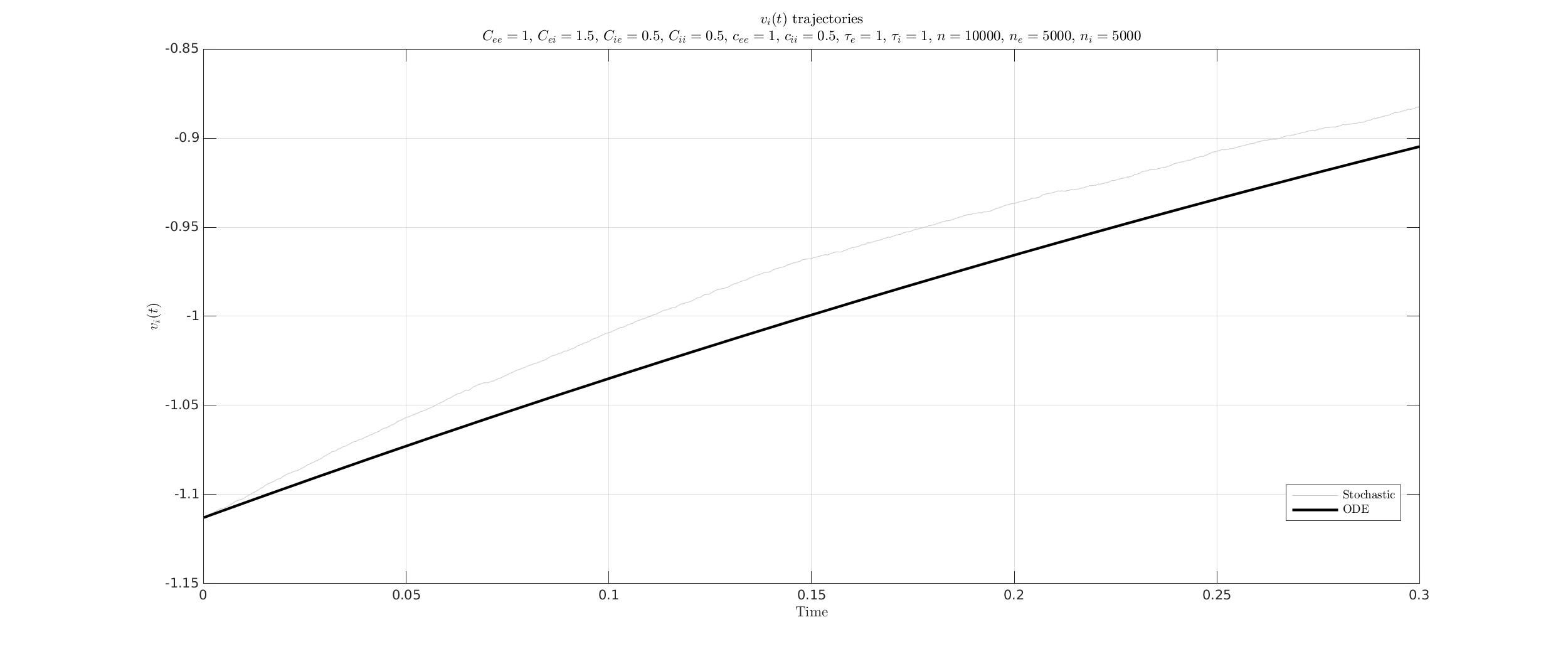}
        {\small $v^n_i(t)$ vs. $\bar{v}_i(t)$}
        \label{fig:comparison_v_i}
    \end{minipage}
    \hfill
    \begin{minipage}{0.49\textwidth}
        \centering
        \includegraphics[width=\textwidth]{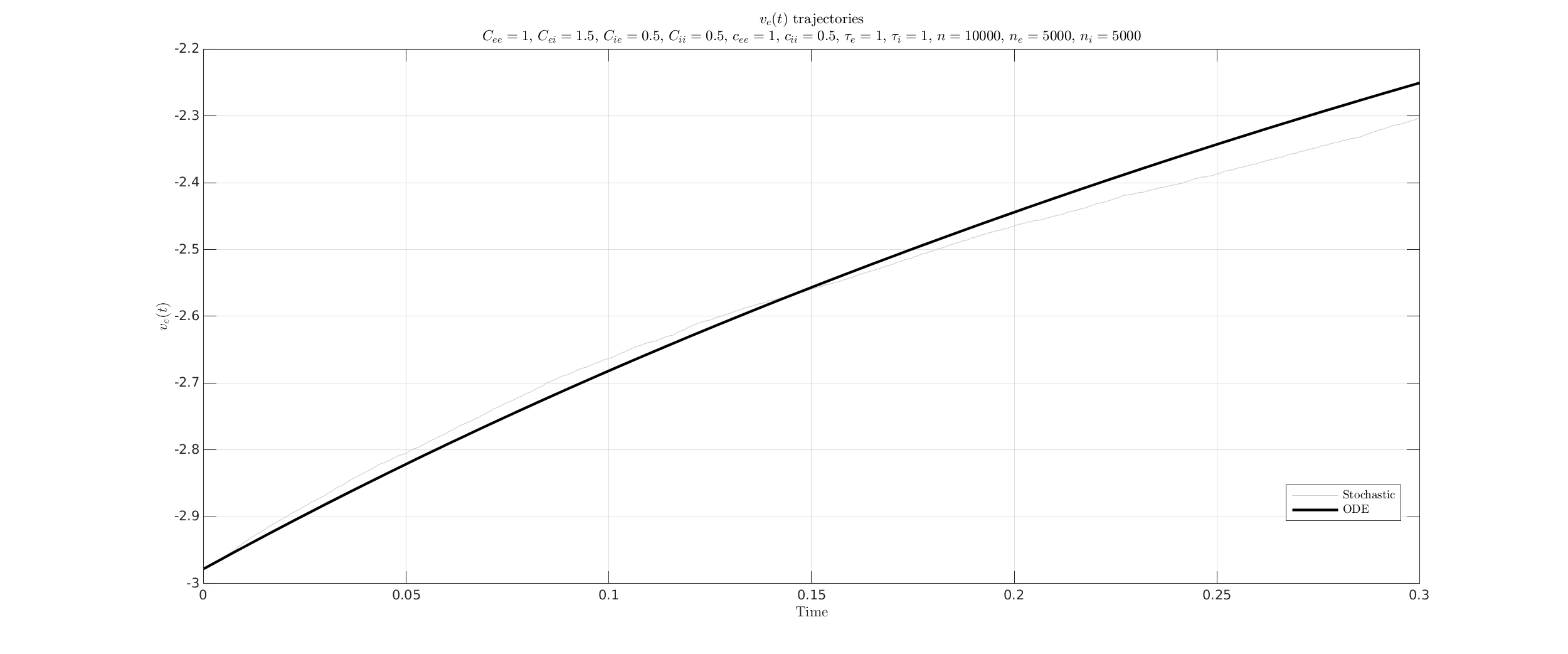}
        {\small $v^n_e(t)$ vs. $\bar{v}_e(t)$}
        \label{fig:comparison_v_e}
    \end{minipage}

    \vspace{0.5cm}

\end{center}

\subsection{Results for $k_e=1$, $k_i=1/2$}

%
%



 \begin{center}   
    \begin{minipage}{0.49\textwidth}
        \centering
        \includegraphics[width=\textwidth]{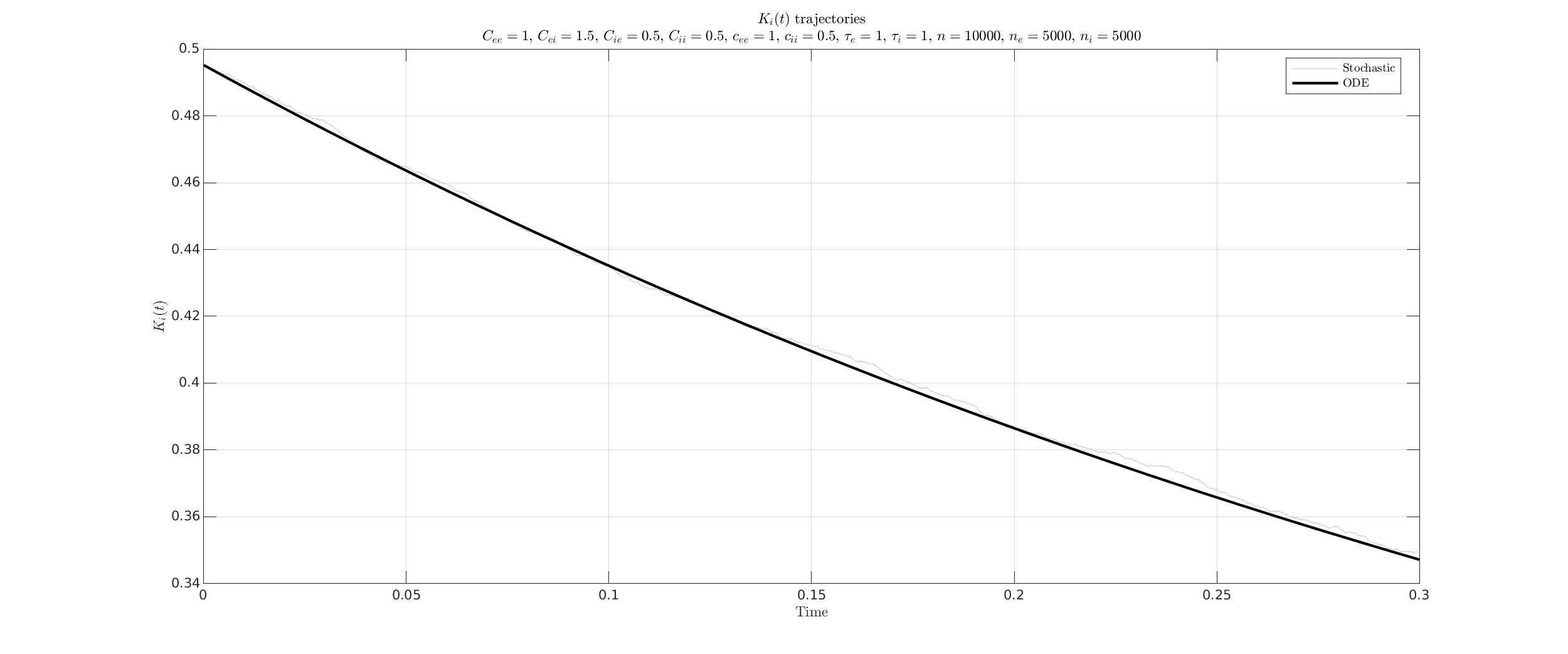}
        {\small $K^n_i(t)$ vs. $K_i(t)$}
    \end{minipage}
    \hfill
    \begin{minipage}{0.49\textwidth}
        \centering
        \includegraphics[width=\textwidth]{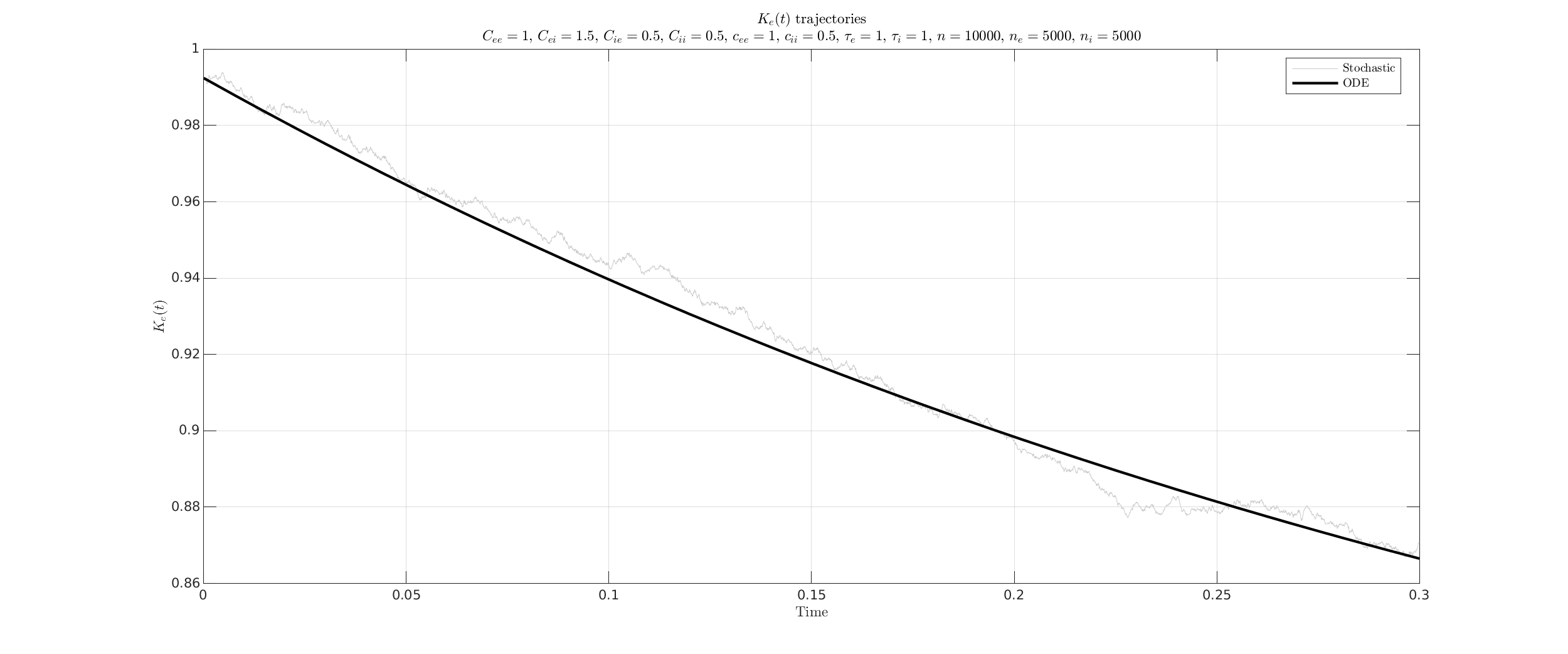}
        {\small $K^n_e(t)$ vs. $K_e(t)$}
    \end{minipage}

    \vspace{0.5cm}

    \begin{minipage}{0.49\textwidth}
        \centering
        \includegraphics[width=\textwidth]{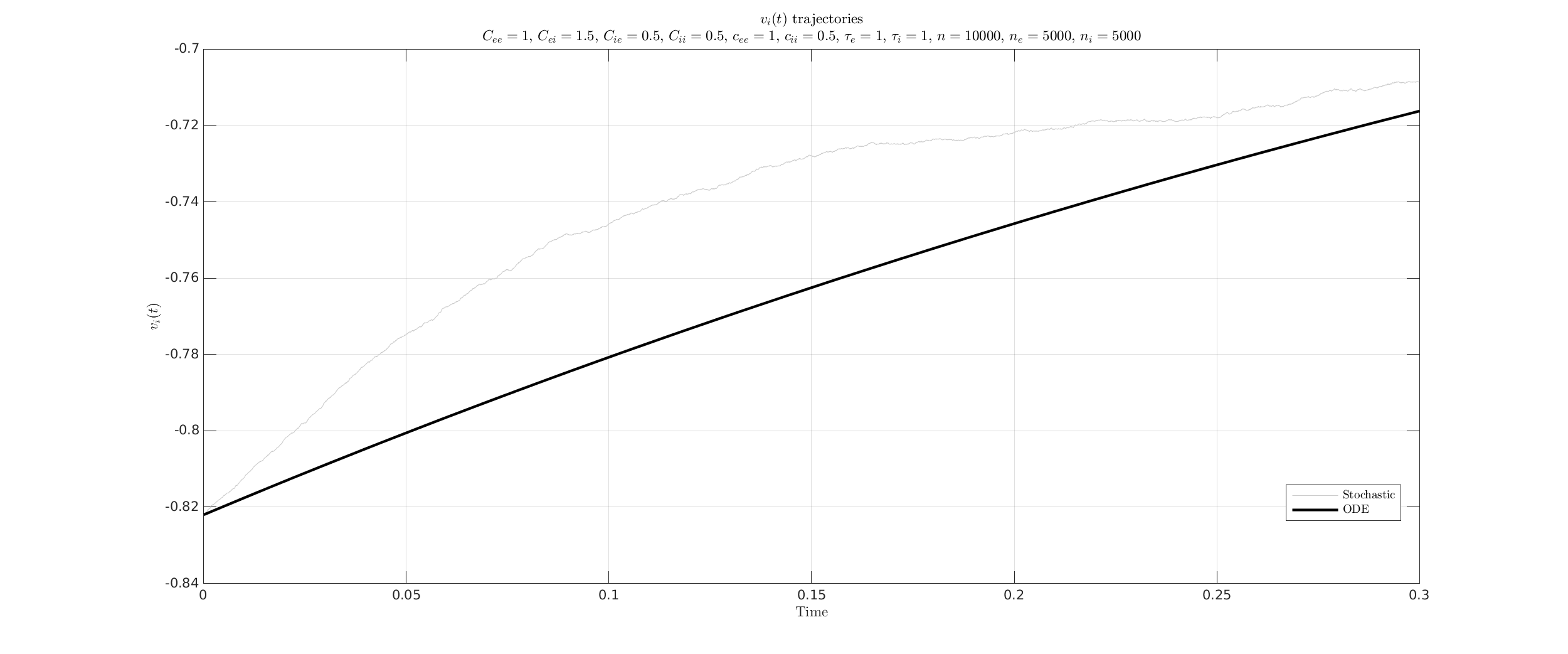}
        {\small $v^n_i(t)$ vs. $\bar{v}_i(t)$}
    \end{minipage}
    \hfill
    \begin{minipage}{0.49\textwidth}
        \centering
        \includegraphics[width=\textwidth]{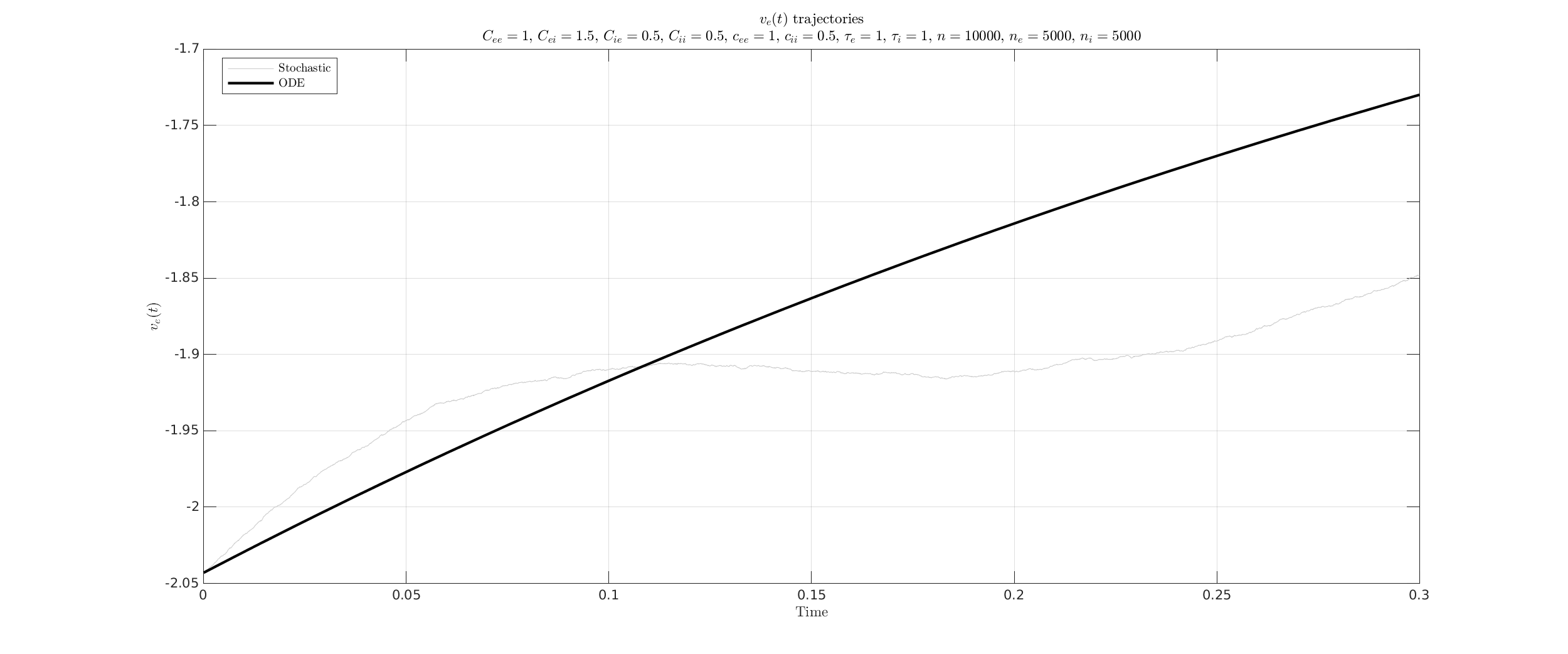}
        {\small $v^n_e(t)$ vs. $\bar{v}_e(t)$}
    \end{minipage}

    \vspace{0.5cm}

\end{center}
\section{Discussion and Conclusion}

We have determined equations that describe the hydrodynamic limit of high-dimensional `balanced' interacting particle systems. Two of the important novelties in our model are that (i) noise is directly afferent on the synaptic connections (as those who postulate that spike-transmission failure is the dominant source of noise in the brain) (ii) the effect of one neuron on another scales as $(O(n^{-1/2}))$, (iii) neurons are either excitatory of inhibitory. Since the landmark work of Sompolinsky and Van Vreeswijk \cite{Vreeswijk1996} this paradigm has been extremely popular in the theoretical neuroscience community, but the literature on rigorous mathematical proofs is comparatively underdeveloped.

An interesting aspect of our results is that the system-wide activity only concentrates (in the large $n$ limit) upto the time that it leaves the balanced manifold. An intriguing further research direction would be to construct an example where this happens and try to understand what will happen afterward. One would expect an abrupt and discontinuous change in the neural activity. Such breakdowns of the balance between excitation and inhibition have certainly been conjectured in the theoretical neuroscience community.

In our numerical simulations, the functions $\lbrace f_{\alpha\beta} \rbrace_{\alpha,\beta}$ are sigmoidal, with a relatively small gradient. This small gradient implies that the convergence to the hydrodynamic limit is relatively slow as $n\to\infty$. Indeed, writing $\bar{\zeta}_t = \inf_{s\leq t} \zeta_s$ (and recalling that $-\zeta_t$ is the real part of the eigenvalues of the Jacobian at the balanced manifold at time $t$, typically $\zeta_t \ll 1$ for sigmoidal firing rate functions). Our proof  demonstrates that
\begin{align}
\big| v^n_e(t) - \bar{v}_e(t) \big| + \big| v^n_i(t) - \bar{v}_i(t) \big| \simeq O\big( n^{-1/2} \bar{\zeta}_t^{-1} \big).
\end{align}
Thus even for large $n \sim 10^4$ (as in our numerical experiments) the stochastic fluctuations are still substantial. In future work, we will determine higher-order approximations that will be accurate for smaller $n$.


\clearpage
\section{Appendix: Bounds on Counting Process Fluctuations}

Let $y(t)$ be a unit-intensity counting process, $n$ a positive integer, and define
\begin{align}
w(t) = y(nt) - nt.
\end{align}
\begin{lemma}
For all $x \geq 0$,
 \begin{align}
\mathbb{P}\big( \sup_{t\leq T} \big| w(t) \big| \geq x \sqrt{n} \big) \leq 2  \exp\bigg(  - \frac{x^2}{4T} \bigg).
\end{align}
\end{lemma}
\begin{proof}
\begin{align}
\mathbb{P}\big( \sup_{t\leq T} \big| w(t) \big| \geq x \sqrt{n} \big) \leq \mathbb{P}\big( \sup_{t\leq T}   w(t)  \geq x \sqrt{n} \big) + \mathbb{P}\big( \inf_{t\leq T}  w(t)   \leq -x \sqrt{n} \big). \label{eq: the above}
\end{align}
Starting with the first term on the RHS of \eqref{eq: the above}, thanks to Doob's submartingale inequality, for any constant $b \in (0,1]$,
\begin{align}
\mathbb{P}\big( \sup_{t\leq T} w(t) \geq x \sqrt{n} \big) \leq & \mathbb{E}\bigg[ \exp\big( bw(T) - b \sqrt{n} x \big) \bigg] \\
= & \exp\bigg( nT \big( \exp(b)-b-1 \big) - b \sqrt{n} x \bigg) \\
\leq & \exp\bigg( b^2 n T - b\sqrt{n} x \bigg) \\
\leq & \exp\bigg(  - \frac{x^2}{4T} \bigg),
\end{align}
after substituting $b = x / (2\sqrt{n}T)$. We analogously find that
\begin{align}
 \mathbb{P}\big( \inf_{t\leq T}  w(t)   \leq -x \sqrt{n} \big) \leq \exp\bigg(  - \frac{x^2}{4T} \bigg)
\end{align}

\end{proof}
\begin{corollary} \label{Corollary Expectation Square}
For any $ b > 0$,
\begin{align}
\lim_{T \to 0^+} \mathbb{E}\bigg[  \exp\bigg( b n^{-1} \sup_{t\leq T} \big| w(t) \big|^2 \bigg) \bigg] = 1.
\end{align}
\end{corollary}

Define the modulus of continuity $\phi_{\epsilon}: \mathcal{C}([0,S] , \mathbb{R}) \mapsto \mathbb{R}$ to be
\begin{align}
\phi_{\epsilon}(w) = \sup_{0\leq s,t \leq S \; : \; |s-t| \leq \epsilon}\big| w(s) - w(t) \big|
\end{align}
\begin{lemma}\label{Lemma Appendix}
For any $\delta > 0$,
\begin{align}
 \lim_{\epsilon \to 0^+} \lsup{n} n^{-1} \log \mathbb{P}\bigg( \sum_{j\in I_n} \phi_{\epsilon}(w^j) \geq \delta n \bigg)  = -\infty.
\end{align}
\end{lemma}
\begin{proof}
Thanks to Chernoff's Inequality, for any $b > 0$
\begin{align}
 \mathbb{P}\bigg( \sum_{j\in I_n} \phi_{\epsilon}(w^j) \geq \delta n \bigg) \leq & \mathbb{E}\bigg[ \exp\bigg( b\sum_{j\in I_n} \phi_{\epsilon}(w^j) - b \delta n \bigg) \bigg] \\
 =& \mathbb{E}\bigg[ \exp\bigg( b \phi_{\epsilon}(w^1) - b \delta \bigg) \bigg]^n
\end{align}
Now for any $b > 0$,
\begin{align}
\lim_{\epsilon \to 0^+} \mathbb{E}\bigg[ \exp\big( b \phi_{\epsilon}(w^1) \big) \bigg] = 1.
\end{align}
Hence by taking $b$ to be arbitrarily large, we obtain the Lemma.
\end{proof}

\bibliographystyle{plain}
\bibliography{library}

\end{document}